\documentclass{article}

 \newif\iflong
\newif\ifshort

 \longtrue

\iflong
\else
\shorttrue
\fi

 \usepackage[preprint]{neurips_2021}

\usepackage{amsmath,todonotes,amsthm,amssymb}

\usepackage{enumitem}
\setitemize{itemsep=0pt}

\usepackage[utf8]{inputenc}  \usepackage[T1]{fontenc}     \usepackage{hyperref}        \usepackage{cleveref}

\usepackage{url}             \usepackage{booktabs}        \usepackage{amsfonts}        \usepackage{nicefrac}        \usepackage{microtype}       \usepackage{xcolor}

 \usepackage{boxedminipage,xspace}

\newcommand{\cc}[1]{{\mbox{\textnormal{\textsf{#1}}}}\xspace}      
\newcommand{\SB}{\{\,}  \newcommand{\SM}{\;{|}\;}  \newcommand{\SE}{\,\}}  
\newcommand{\Nat}{\mathbb{N}}

\newcommand{\NP}{\cc{NP}}
\newcommand{\coNP}{\cc{co-NP}}
\newcommand{\FPT}{\cc{FPT}}
\newcommand{\XP}{\cc{XP}}
\newcommand{\W}{{\cc{W}}}

\newcommand{\YES}{\cc{Yes}}
\newcommand{\NO}{\cc{No}}

\newcommand{\tw}{\operatorname{tw}}
\newcommand{\tcw}{\operatorname{tcw}}
\newcommand{\cut}{\operatorname{cut}}
\newcommand{\adh}{\operatorname{adh}}
\newcommand{\tor}{\operatorname{tor}}
\newcommand{\loc}{\operatorname{loc}}
\newcommand{\con}{\operatorname{con}}
\newcommand{\inn}{\operatorname{inn}}
\newcommand{\fes}{\ensuremath{\operatorname{fen}}}
\newcommand{\lfes}{\ensuremath{\operatorname{lfen}}}

\newcommand{\bigoh}{\mathcal{O}}

\newcommand{\PPP}{\mathcal{P}}
\newcommand{\III}{\mathcal{I}}
\newcommand{\QQQ}{\mathcal{Q}}

\newcommand{\FFF}{\mathcal{F}}

\newcommand{\XXX}{\mathcal{X}}
\newcommand{\RRR}{\mathcal{R}}
\newcommand{\DDD}{\Upsilon}
\newcommand{\PT}{\Psi}

\newcommand{\MFAS}{\textsc{MFAS}}
\newcommand{\RMC}{\textsc{RMC}}
\newcommand{\BNSL}{\textsc{BNSL}}
\newcommand{\UDMSS}{\textsc{UDMSS}}
\newcommand{\DMSS}{\textsc{DMSS}}
\newcommand{\MSS}{\textsc{MSS}}
\newcommand{\BNSLneq}{\textsc{BNSL}$^{\neq0}$}
\newcommand{\BNSLadd}{\textsc{BNSL}$^+$}
\newcommand{\BNSLaddq}{\textsc{BNSL}$^+_{\leq}$}
\newcommand{\PL}{\textsc{PL}}
\newcommand{\PLneq}{\textsc{PL}$^{\neq0}$}
\newcommand{\PLadd}{\textsc{PL}$^+$}
\newcommand{\PLaddq}{\textsc{PL}$^+_{\leq}$}
\newcommand{\score}{\texttt{score}}
  \newcommand{\parentsets}{\Gamma_f}
\newcommand{\inneighb}{P_{\rightarrow}}
\newcommand{\Con}{\texttt{\textup{Con}}}
\newcommand{\trcl}{\texttt{\textup{trcl}}}
\newcommand{\opt}{\texttt{\textup{opt}}}

\newtheorem{theorem}{Theorem}
\newtheorem{claim}{Claim}
\newtheorem{observation}[theorem]{Observation}
\newtheorem{proposition}[theorem]{Proposition}

\newtheorem{lemma}[theorem]{Lemma}
\newtheorem{redrule}{Reduction Rule}

\title{The Complexity of Bayesian Network Learning: Revisiting the Superstructure}

\author{Robert Ganian and Viktoriia Korchemna\\
Algorithms and Complexity Group, TU Wien\\
\texttt{\{rganian,vkorchemna\}@ac.tuwien.ac.at}
                                                                  }

\begin{document}

\maketitle

\begin{abstract}
We investigate the parameterized complexity of Bayesian Network Structure Learning (BNSL), a classical problem that has received significant attention in empirical but also purely theoretical studies. We follow up on previous works that have analyzed the complexity of BNSL w.r.t.\ the so-called \emph{superstructure} of the input. While known results imply that BNSL is unlikely to be fixed-parameter tractable even when parameterized by the size of a vertex cover in the superstructure, here we show that a different kind of parameterization---notably by the size of a feedback edge set---yields fixed-parameter tractability. We proceed by showing that this result can be strengthened to a localized version of the feedback edge set, and provide corresponding lower bounds that complement previous results to provide a complexity classification of BNSL w.r.t.\ virtually all well-studied graph parameters.

We then analyze how the complexity of BNSL depends on the representation of the input. In particular, while the bulk of past theoretical work on the topic assumed the use of the so-called \emph{non-zero representation}, here we prove that if an \emph{additive representation} can be used instead then BNSL becomes fixed-parameter tractable even under significantly milder restrictions to the superstructure, notably when parameterized by the treewidth alone. Last but not least, we show how our results can be extended to the closely related problem of Polytree Learning.
\end{abstract}

\section{Introduction}

Bayesian networks are among the most prominent graphical models for probability distributions. The key feature of Bayesian networks is that they represent conditional dependencies between random variables via a directed acyclic graph; the vertices of this graph are the variables, and an arc $ab$ means that the distribution of variable $b$ depends on the value of $a$. One beneficial property of Bayesian networks is that they can be used to infer the distribution of random variables in the network based on the values of the remaining variables.

The problem of constructing a Bayesian network with an optimal network structure is \NP-hard, and remains \NP-hard even on highly restricted instances~\cite{Chickering95}. This initial negative result has prompted an extensive investigation of the problem's complexity, with the aim of identifying new tractable fragments as well as the boundaries of its intractability~\cite{KorhonenP13,OrdyniakS13,KorhonenP15,GruttemeierK20,ElidanG08,Dasgupta99,GaspersKLOS15}. The problem---which we simply call \textsc{Bayesian Network Structure Learning} (\BNSL)---can be stated as follows: given a set of $V$ of variables  (represented as vertices), a family $\FFF$ of \emph{score functions} which assign each variable $v\in V$ a score based on its \emph{parents}, and a target value $\ell$, determine if there exists a directed acyclic graph over $V$ that achieves a total score of at least $\ell$\footnote{Formal definitions are provided in Section~\ref{sec:prelims}. We consider the decision version of \BNSL\ for complexity-theoretic reasons only; all of the provided algorithms are constructive and can output a network as a witness.}. 

To obtain a more refined understanding of the complexity of \BNSL, past works have analyzed the problem not only in terms of classical complexity but also from the perspective of \emph{parameterized complexity}~\cite{DowneyFellows13,CyganFKLMPPS15}. In parameterized complexity analysis, the tractability of problems is measured with respect to the input size $n$ and additionally with respect to a specified numerical \emph{parameter} $k$. In particular, a problem that is \NP-hard in the classical sense may---depending on the parameterization used---be \emph{fixed-parameter tractable} (\FPT), which is the parameterized analogue of polynomial-time tractability and means that a solution can be found in time $f(k)\cdot n^{\bigoh(1)}$ for some computable function $f$, or \W$[1]$-\emph{hard}, which rules out fixed-parameter tractability under standard complexity assumptions. The use of parameterized complexity as a refinement of classical complexity is becoming increasingly common and has been employed not only for \BNSL~\cite{KorhonenP13,OrdyniakS13,KorhonenP15}, but also for numerous other problems arising in the context of neural networks and artificial intelligence~\cite{GanianKOS18,SimonovFGP19,EibenGKS19,GanianO18}.

Unfortunately, past complexity-theoretic works have shown that \BNSL\ is a surprisingly difficult problem. In particular, not only is the problem \NP-hard, but it remains \NP-hard even when asking for the existence of extremely simple networks such as directed paths~\cite{Meek01} and is \W$[1]$-hard when parameterized by the \emph{vertex cover number} of the network~\cite{KorhonenP15}. In an effort to circumvent these lower bounds, several works have proposed to
 instead consider restrictions to the so-called \emph{superstructure}, which is a graph that, informally speaking, captures all potential dependencies between variables~\cite{TsamardinosBA06,PerrierIM08}. Ordyniak and Szeider~\cite{OrdyniakS13} studied the complexity of \BNSL\ when parameterized by the structural properties of the superstructure, and showed that parameterizing by the \emph{treewidth}~\cite{RobertsonSeymour86} of the superstructure is sufficient to achieve a weaker notion of tractability called \XP-\emph{tractability}. However, they also proved that \BNSL\ remains \W$[1]$-hard when parameterized by the treewidth of the superstructure~\cite[Theorem 3]{OrdyniakS13}. 

 \noindent
\textbf{Contribution.}\quad
Up to now, no ``implicit'' restrictions of the superstructure were known to lead to a fixed-parameter algorithm for \BNSL\ alone. More precisely, the only known fixed-parameter algorithms for the problem require that we place explicit restrictions on either the sought-after network or the parent sets on the input: \BNSL\ is known to be fixed-parameter tractable when parameterized by the number of arcs in the target network~\cite{GruttemeierK20}, the treewidth of an ``\emph{extended superstructure graph}'' which also bounds the maximum number of parents a variable can have~\cite{KorhonenP13}, or
 the number of parent set candidates plus the treewidth of the superstructure~\cite{OrdyniakS13}.
     Moreover,
a closer analysis of the reduction given by Ordyniak and Szeider~\cite[Theorem 3]{OrdyniakS13} reveals that \BNSL\ is also \W$[1]$-hard when parameterized by the \emph{treedepth}, \emph{pathwidth}, and even the \emph{vertex cover number} of the superstructure alone. 
The vertex cover number is equal to the vertex deletion distance to an edgeless graph, and hence their result essentially rules out the use of the vast majority of graph parameters; among others, any structural parameter based on vertex deletion distance.

As our first conceptual contribution, we show that a different kind of graph parameters---notably, parameters that are based on edge deletion distance---give rise to fixed-parameter algorithms for \BNSL\ in its full generality, without requiring any further explicit restrictions on the target network or parent sets. Our first result in this direction concerns the \emph{feedback edge number} (\fes), which is the minimum number of edges that need to be deleted to achieve acyclicity. In Theorem~\ref{thm:kernel} we show not only that \BNSL\ is fixed-parameter tractable when parameterized by the \fes\ of the superstructure, but also provide a polynomial-time preprocessing algorithm that reduces any instance of \BNSL\ to an equivalent one whose number of variables is linear in the \fes\ (i.e., a \emph{kernelization}~\cite{DowneyFellows13,CyganFKLMPPS15}).

    Since \fes\ is a highly ``restrictive'' parameter---its value can be large even on simple superstructures such as collections of disjoint cycles---we proceed by asking whether it is possible to lift fixed-parameter tractability to a more relaxed way of measuring distance to acyclicity. For our second result, we introduce the \emph{local feedback edge number} (\lfes), which 
intuitively measures the maximum edge deletion distance to acyclicity for cycles intersecting any particular vertex in the superstructure.
 In Theorem~\ref{thm:lfes}, we show that \BNSL\ is also fixed-parameter tractable when prameterized by \lfes; we also show that this comes at the cost of \BNSL\ not admitting any polynomial-time preprocessing procedure akin to Theorem~\ref{thm:kernel} when parameterized by \lfes. We conclude our investigation in the direction of parameters based on edge deletion distance by showing that \BNSL\ parameterized by \emph{treecut width}~\cite{MarxW14,Wollan15,Ganian0S15}, a recently discovered edge-cut based counterpart to treewidth, remains \W$[1]$-hard (Theorem~\ref{thm:tcwhard}). An overview of these complexity-theoretic results is provided in Figure~\ref{fig:overview}.

\begin{figure}[htb]
\begin{minipage}[c]{0.5\textwidth}
\vspace{-0.1cm}
\includegraphics[width=\textwidth]{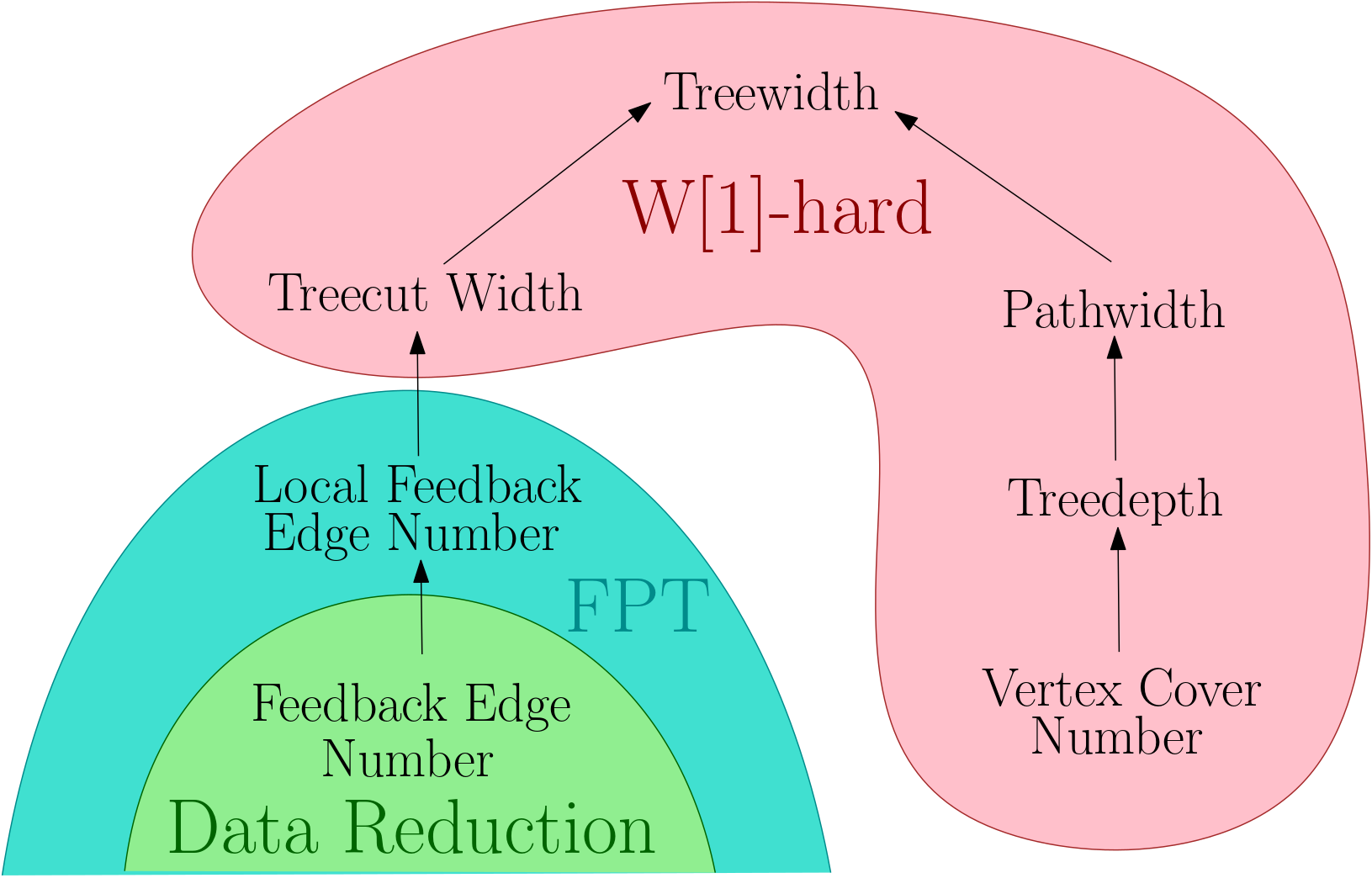}\vspace{-0.6cm}
\end{minipage}
\hfill
\begin{minipage}[c]{0.45\textwidth}
\caption{The complexity landscape of \BNSL\ with respect to parameterizations of the superstructure. Arrows point from more restrictive parameters to less restrictive ones. Results for the three graph parameters on the left side follow from this paper, while all other \W$[1]$-hardness results follow from the reduction by Ordyniak and Szeider~\cite[Theorem 3]{OrdyniakS13}.}
\label{fig:overview}
\end{minipage}
       \end{figure}

As our second conceptual contribution, we show that \BNSL\ becomes significantly easier when one can use an \emph{additive representation} of the scores rather than the \emph{non-zero representation} that was considered in the vast majority of complexity-theoretic works on \BNSL\ to date~\cite{KorhonenP13,OrdyniakS13,KorhonenP15,GruttemeierK20,ElidanG08,GaspersKLOS15}. The additive representation is inspired by known heuristics for \BNSL~\cite{ScanagattaCCZ15,ScanagattaCCZ16} and utilizes a succinct encoding of the score function which assumes that the scores for parent sets can be decomposed into a sum of the scores of individual variables in the parent set; a discussion and formal definitions are provided in Section~\ref{sec:prelims}. In Theorem~\ref{thm:additive}, we show that if the additive representation can be used, \BNSL\ becomes fixed-parameter tractable when parameterized by the treewidth of the superstructure (and hence under every parameterization depicted in Figure~\ref{fig:overview}). Motivated by the empirical usage of the additive representation, we also consider the case where we additionally impose a bound $q$ on the number of parents a vertex can accept; we show that the result of Theorem~\ref{thm:additive} also covers this case if $q$ is taken as an additional parameter, and otherwise rule out fixed-parameter tractability using an intricate reduction (Theorem~\ref{thm:addwhard}).
 
For our third and final conceptual contribution, we show how our results can be adapted for the emergent problem of \textsc{Polytree Learning} (\PL), a variant of \BNSL\ where we require that the network forms a polytree. The crucial advantage of such networks is that they allow for a more efficient solution of the inference task~\cite{PearlBook,GuoHsu02}, and the complexity of \PL\ has been studied in several works~\cite{GKM21,GaspersKLOS15,SafaeiMS13}. We show that all our results for \BNSL\ can be adapted to \PL, albeit in some cases it is necessary to perform non-trivial modifications.
 Furthermore, we observe that unlike \BNSL, \PL\ becomes polynomial-time tractable when the additive representation is used (Observation~\ref{obs:PLeasy}); this matches the ``naive'' expectation that learning simple networks would be easier than \BNSL\ in its full generality. As our concluding result, we show that this expectation is in fact not always validated: while \PL\ was recently shown to be \W$[1]$-hard when parameterized by the number of so-called \emph{dependent vertices}~\cite{GKM21}, in Theorem~\ref{thm:depfpt} we prove that \BNSL\ is fixed-parameter tractable under that same parameterization.

\section{Preliminaries}
\label{sec:prelims}

For an integer $i$, we let $[i]=\{1,2,\dots,i\}$ and $[i]_0=[i]\cup\{0\}$.
We denote by $\Nat$ the set of natural numbers, by $\Nat_0$ the set
$\Nat \cup \{0\}$. 
 \ifshort
We refer to the handbook by Diestel~\cite{Diestel12} for
standard terminology on directed as well as undirected graphs.
  The \emph{skeleton} (sometimes called the \emph{underlying undirected graph}) of a directed graph (a \emph{digraph}) $D=(V,A)$ is the undirected graph $G'=(V,E)$  such that $vw\in E$ if $vw\in A$ or $wv\in A$. A digraph is a \emph{polytree} if its skeleton is a forest.

When comparing two numerical parameters $\alpha, \beta$ of graphs, we say that $\alpha$ is more \emph{restrictive} than $\beta$ if there exists a function $f$ such that $\beta(G)\leq f(\alpha(G))$ holds for every graph $G$. 
We refer to the standard sources for the fundamentals of parameterized complexity,
 including the definitions of \emph{fixed-parameter tractability}, \emph{parameterized reductions}, \W$[1]$-\emph{hardness} and \emph{treewidth}~\cite{CyganFKLMPPS15,DowneyFellows13,Niedermeier06}. 
 \fi

\iflong
We refer to the handbook by Diestel~\cite{Diestel12} for
standard graph terminology. In this paper, we will consider directed as well as undirected simple graphs. If $G=(V,E)$ is an undirected graph and $\{v,w\}\in E$, we will often use $vw$ as shorthand for $\{v,w\}$; we will also sometimes use $V(G)$ to denote its vertex set. Moreover, we let $N_G(v)$ denote the set of \emph{neighbors} of $v$, i.e., $\SB u\in V\SM vu \in E\SE$. We extend this notation to sets as follows: $N_G(X)=\SB u\in V\setminus X \SM \exists x\in X: ux\in E(G)$. For a set $X$ of vertices, let $A_X$ denote the set of all possible arcs over $X$.

If $D=(V,A)$ is a directed graph (i.e., a \emph{digraph}) and $(v,w)\in A$, we will similarly use $vw$ as shorthand for $(v,w)$. We also let $P_D(v)$ denote the set of \emph{parents} of $v$, i.e., $\SB u\in V\SM uv \in A\SE$ (there are sometimes called \emph{in-neighbors} in the literature, while the notion of \emph{out-neighbors} is defined analogously). In both cases, we may drop $G$ or $D$ from the subscript if the (di)graph is clear from the context. The \emph{degree} of $v$ is $|N(v)|$, and for digraphs we use the notions of \emph{in-degree} (which is equal to $|P(v)|$) and \emph{out-degree} (the number of arcs originating from the given vertex).

 The \emph{skeleton} (sometimes called the \emph{underlying undirected graph}) of a digraph $G=(V,A)$ is the undirected graph $G'=(V,E)$  such that $vw\in E$ if $vw\in A$ or $wv\in A$. A digraph is a \emph{polytree} if its skeleton is a forest.
      
When comparing two numerical parameters $\alpha, \beta$ of graphs, we say that $\alpha$ is more \emph{restrictive} than $\beta$ if there exists a function $f$ such that $\beta(G)\leq f(\alpha(G))$ holds for every graph $G$. In other words, $\alpha$ is more restrictive than $\beta$ if and only if the following holds: whenever all graphs in some graph class $\mathcal{H}$ have $\alpha$ upper-bounded by a constant, all graphs in $\mathcal{H}$ also have $\beta$ upper-bounded by a constant. Observe that in this case a fixed-parameter algorithm parameterized by $\beta$ immediately implies a fixed-parameter algorithm parameterized by $\alpha$, while \W$[1]$-hardness behaves in the opposite way.
\fi

 \noindent \textbf{Problem Definitions.}\quad
Let $V$ be a set of vertices and $\FFF=\SB f_v: 2^{V\setminus \{v\}}\rightarrow \Nat_0 \SM v\in V\SE$ be a family of \emph{local score functions}. 
 For a digraph $D=(V,A)$, we define its score as follows: $\score(D)=\sum_{v\in V}f_v(P_D(v))$, where $P_D(v)$ is the set of vertices of $D$ with an outgoing arc to $v$ (i.e., the \emph{parent set} of $v$ in $D$). We can now formalize our problem of interest~\cite{OrdyniakS13,GruttemeierK20}.
  
\noindent
\begin{center}
\begin{boxedminipage}{0.98 \columnwidth}
\textsc{Bayesian Network Structure Learning} (\BNSL)\\[5pt]
\begin{tabular}{l p{0.83 \columnwidth}}
Input: & A set $V$ of vertices, a family $\FFF$ of local score functions, and an integer $\ell$.\\
Question: & Does there exist an acyclic digraph $D=(V,A)$ such that $\score(D)\geq \ell$?
\end{tabular}
\end{boxedminipage}
\end{center}

\textsc{Polytree Learning} (\PL) is defined analogously, with the only difference that there $D$ is additionally required to be a polytree~\cite{GKM21}. We call $D$ a \emph{solution} for the given instance.

Since both $V$ and $\FFF$ are assumed to be given on the input of our problems, an issue that arises here is that an explicit representation of $\FFF$ would be exponentially larger than $|V|$. A common way to potentially circumvent this is to use a \emph{non-zero representation} of the family $\FFF$, i.e., where we only store values for $f_v(P)$ that are different than zero. 
\iflong 
This model has been used in a large number of works studying the complexity of \BNSL\ and \PL~\cite{KorhonenP13,OrdyniakS13,KorhonenP15,GruttemeierK20,GaspersKLOS15,GKM21} and is known to be strictly more general than, e.g., the bounded-arity representation where one only considers parent sets of arity bounded by a constant~\cite[Section 3]{OrdyniakS13}. 
\fi
\ifshort
This model has been used in the vast majority of works studying the complexity of \BNSL\ and \PL~\cite{KorhonenP13,OrdyniakS13,KorhonenP15,GruttemeierK20,GaspersKLOS15,GKM21}.
\fi
Let $\parentsets(v)$ be the collection of candidate parents of $v$ which yield a non-zero score; formally, $\parentsets(v)=\SB Z \SM f_v(Z)\neq 0\SE$, and the input size $|\III|$ of an instance $\III=(V,\FFF,\ell)$ is simply defined as $|V|+\ell+\sum_{v\in V, P\in \parentsets(v)}|P|$\footnote{
We remark that the non-zero representation could be strengthened even further by omitting each parent set $Z$ of $v$ which contains a proper subset $Z'$ such that $f_v(Z)\leq f_v(Z')$. This preprocessing step, however, does not have an impact on any of the results presented in this paper.}.

\ifshort

A natural way to think about and exploit the structure of inter-variable dependencies of an instance $\III$ is to consider its \emph{superstructure graph} $G_\III=(V,E)$, where $ab\in E$ if $a$ occurs in at least one candidate set in $\Gamma_f(b)$ (or vice-versa).
\fi
\iflong
Let $\inneighb(v)$ be the set of all parents which appear in $\parentsets(v)$, i.e., $a\in \inneighb(v)$ if and only if $\exists Z\in \parentsets(v): a\in Z$.
 A natural way to think about and exploit the structure of inter-variable dependencies laid bare by the non-zero representation is to consider the \emph{superstructure graph} $G_\III=(V,E)$ of a \BNSL\ (or \PL) instance $\III=(V,\FFF,\ell)$, where $ab\in E$ if and only if either $a\in \inneighb(b)$, or $b\in \inneighb(a)$, or both.
\fi
An example is provided in Figure~\ref{fig:superstructure}.

\begin{figure}
\begin{minipage}[c]{0.4\textwidth}
\includegraphics[width=\textwidth]{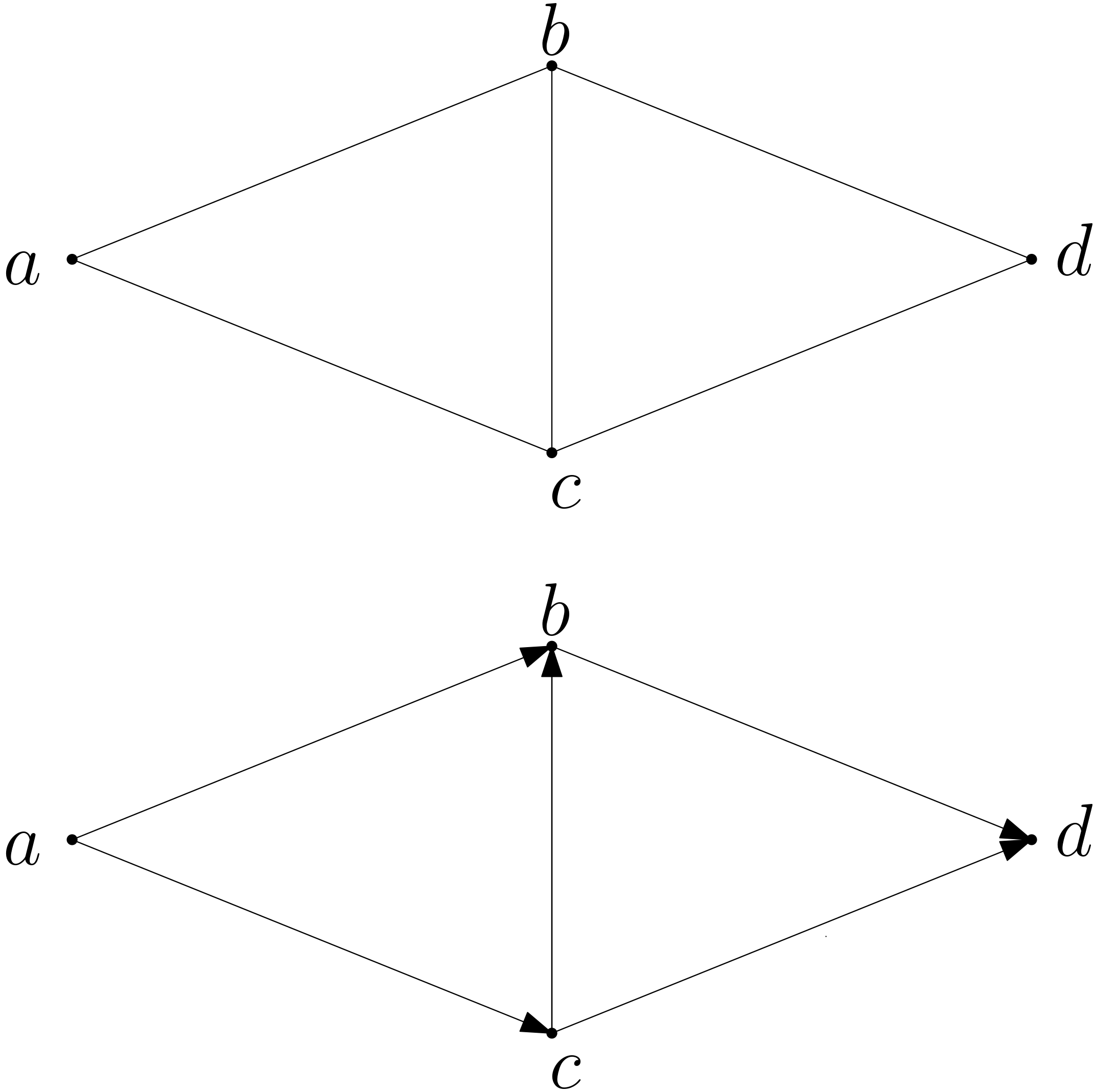}
\end{minipage}
\hfill
\begin{minipage}[c]{0.5\textwidth}
\caption{Example of a superstructure graph (on the top) and a suitable solution DAG (on the bottom) when:\\
$f_a (\{b\})=f_a (\{c\})=1$, $f_a(\{b,c\})=2$;\\
$f_b (\{a\})=f_b (\{c\})=1$, $f_b(\{a,c\})=3$;\\
$f_c (\{a\})=3$, $f_c(\{b\})=2$;\\
$f_d(\{b,c\})=1$;\\
$l=6.$\\
Note that in the depicted DAG the scores of $b$ and $c$ are equal to 3, the score of $d$ is equal to 1. Parent set of $a$ is empty; as we assume the non-zero representation and $f_a(\emptyset)$ is not specified in the input, we conclude that $f_a(\emptyset)=0$. Therefore the total score is $0+3+3+1=7\geq6=l$.}

\label{fig:superstructure}
\end{minipage}
\end{figure}

Naturally, families of local score functions may be exponentially larger than $|V|$ even when stored using the non-zero representation. In this paper, we also consider a second representation of $\FFF$ which is guaranteed to be polynomial in $|V|$: in the \emph{additive representation}, we require that for every vertex $v\in V$ and set $Q=\{q_1,\dots,q_m\}\subseteq V\setminus \{v\}$, $f_v(Q)=f_v(\{q_1\})+\dots+f_v(\{q_m\})$. Hence, each cost function $f_v$ can be fully characterized by storing at most $|V|$-many entries of the form $f_v(x):=f_v(\{x\})$ for each $x\in V\setminus \{v\}$. To avoid overfitting, one may optionally impose an additional constraint: an upper bound $q$ on the size of any parent set in the solution\iflong (or, equivalently, $q$ is a maximum upper-bound on the in-degree of the sought-after acyclic digraph $D$)\fi.

While not every family of local score functions admits an additive representation, the additive model is similar in spirit to the models used by some practical algorithms for \BNSL. For instance, the algorithms of Scanagatta, de Campos, Corani and Zaffalon~\cite{ScanagattaCCZ15,ScanagattaCCZ16}, which can process \BNSL\ instances with up to thousands of variables, approximate the real score functions by adding up the known score functions for two parts of the parent set and applying a small, logarithmic correction. 
Both of these algorithms also use the aforementioned bound $q$ for the parent set size. In spite of this connection to practice and the representation's streamlined nature, we are not aware of any prior works that considered the additive representation in complexity-theoretic studies of \BNSL\ and \PL.
\ifshort
The superstructure graph for the additive representation can be defined in an analogous way as for the non-zero representation: an edge $uv$ simply captures a ``suspected dependencey' between variables $u$ and $v$ (i.e., one receives a positive score for depending on the other). 
\fi

\iflong
As before, in the additive representation we will also only store scores for parents of $v$ which yield a non-zero score, and can thus define $\inneighb(v)=\SB z \SM f_v(z)\neq 0 \SE$, as for the non-zero representation. This in turn allows us to define the superstructure graphs in an analogous way as before: $G_\III=(V,E)$ where $ab\in E$ if and only if $a\in \inneighb(b)$, $b\in \inneighb(a)$, or both.
\fi

To distinguish between these models, we use \BNSLneq, \BNSLadd, and \BNSLaddq  to denote \textsc{Bayesian Network Structure Learning} with the non-zero representation, the additive representation, and the additive representation and the parent set size bound $q$, respectively\ifshort~(and analogously for \PL).\fi
\iflong. The same notation will also be used for \textsc{Polytree Learning}---for example, an instance of $\PL^+_{\leq}$ will consist of $V$, a family $\FFF$ of local score functions in the additive representation, and integers $\ell$, $q$, and the question is whether there exists a polytree $D=(V,A)$ with in-degree at most $q$ and $\score(D)\geq \ell$.
\fi

In our algorithmic results, we will often use $G=(V,E)$ to denote the superstructure graph of the input instance $\III$. Without any loss of generality, we will also assume that $G$ is connected.
\iflong
Indeed, given an algorithm $\mathbb{A}$ that solves $\BNSL$ on connected instances, we may solve disconnected instances of $\BNSL$ by using $\mathbb{A}$ to find the maximum score $\ell_C$ for each connected component $C$ of $G$ independently, and we may then simply compare $\sum_{C\text{ is a connected component of }G}\ell_C$ with $\ell$.
\fi

\iflong
\smallskip
\noindent \textbf{Parameterized Complexity.}\quad
In parameterized
algorithmics~\cite{CyganFKLMPPS15,DowneyFellows13,Niedermeier06} the
running-time of an algorithm is studied with respect to a parameter
$k\in\Nat_0$ and input size~$n$. The basic idea is to find a parameter
that describes the structure of the instance such that the
combinatorial explosion can be confined to this parameter. In this
respect, the most favorable complexity class is \FPT
(\textit{fixed-parameter tractable}) which contains all problems that
can be decided by an algorithm running in time $f(k)\cdot
n^{\bigoh(1)}$, where $f$ is a computable function. Algorithms with
this running-time are called \emph{fixed-parameter algorithms}. A less
favorable outcome is an \XP{} \emph{algorithm}, which is an algorithm
running in time $\bigoh(n^{f(k)})$; problems admitting such
algorithms belong to the class \XP. 

Showing that a problem is $\W[1]$-hard rules out the existence of a fixed-parameter algorithm under the well-established assumption that $\W[1]\neq \FPT$. This is usually done via a \emph{parameterized reduction}~\cite{CyganFKLMPPS15,DowneyFellows13} to some known $\W[1]$-hard problem. A parameterized reduction from a parameterized problem $\PPP$ to a parameterized problem $\QQQ$ is a function:
\begin{itemize}
\item which maps \YES-instances to \YES-instances and \NO-instances to \NO-instances,
\item which can be computed in time $f(k)\cdot
n^{\bigoh(1)}$, where $f$ is a computable function, and
\item where the parameter of the output instance can be upper-bounded by some function of the parameter of the input instance.
\end{itemize}

\smallskip
\noindent \textbf{Treewidth.}\quad
A \emph{nice tree-decomposition}~$\mathcal{T}$ of a graph $G=(V,E)$ is a pair 
$(T,\chi)$, where $T$ is a tree (whose vertices we call \emph{nodes}) rooted at a node $r$ and $\chi$ is a function that assigns each node $t$ a set $\chi(t) \subseteq V$ such that the following holds: 
\begin{itemize}
	\item For every $uv \in E$ there is a node
	$t$ such that $u,v\in \chi(t)$.
	\item For every vertex $v \in V$,
	the set of nodes $t$ satisfying $v\in \chi(t)$ forms a subtree of~$T$.
	\item $|\chi(\ell)|=1$ for every leaf $\ell$ of $T$ and $|\chi(r)|=0$.
	\item There are only three kinds of non-leaf nodes in $T$:
	\begin{itemize}
	         \item \textbf{Introduce node:} a node $t$ with exactly
          one child $t'$ such that $\chi(t)=\chi(t')\cup
          \{v\}$ for some vertex $v\not\in \chi(t')$.
        \item \textbf{Forget node:} a node $t$ with exactly
          one child $t'$ such that $\chi(t)=\chi(t')\setminus
          \{v\}$ for some vertex $v\in \chi(t')$.
        \item \textbf{Join node:} a node $t$ with two children $t_1$,
          $t_2$ such that $\chi(t)=\chi(t_1)=\chi(t_2)$.
	\end{itemize}
\end{itemize}

The \emph{width} of a nice tree-decomposition $(T,\chi)$ is the size of a largest set $\chi(t)$ minus~$1$, and the \emph{treewidth} of the graph $G$,
denoted $\tw(G)$, is the minimum width of a nice tree-decomposition of~$G$.
 Fixed-parameter algorithms are known for computing a nice tree-decomposition of optimal width~\cite{Bodlaender96,Kloks94}. 
      For $t \in V(T)$ we denote by $T_t$ the subtree of $T$ rooted at $t$. 
\fi

 \noindent \textbf{Graph Parameters Based on Edge Cuts.}\quad
Traditionally, the bulk of graph-theoretic research on structural parameters has focused on parameters that guarantee the existence of small vertex separators in the graph; these are  inherently tied to the theory of \emph{graph minors}~\cite{RobertsonS83,RobertsonSeymour86} and the vertex deletion distance. This approach gives rise not only to the classical notion of treewidth, but also to its well-known restrictions and refinements such as \emph{pathwidth}~\cite{RobertsonS83}, \emph{treedepth}~\cite{NesetrilOssonademendez12} and the \emph{vertex cover number}~\cite{FialaGK11,KobayashiT16}. The vertex cover number is the most restrictive parameter in this hierarchy.
 
However, there are numerous problems of interest that remain intractable even when parameterized by the vertex cover number. A recent approach developed for attacking such problems has been to consider parameters that guarantee the existence of small edge cuts in the graph; these are typically based on the edge deletion distance or, more broadly, tied to the theory of \emph{graph immersions}~\cite{Wollan15,MarxW14}. The parameter of choice for the latter is \emph{treecut width} ($\tcw$)~\cite{Wollan15,MarxW14,Ganian0S15,GanianKO21}, a counterpart to treewidth which has been successfully used to tackle some problems that remained intractable when parameterized by the vertex cover number~\cite{GanianO21}. \iflong For the purposes of this manuscript, it will be useful to note that graphs containing a vertex cover $X$ such that every vertex outside of $X$ has degree at most $2$ have treecut width at most $|X|$~\cite[Section 3]{GanianO21}.\fi

On the other hand, the by far most prominent parameter based on edge deletion distance is the \emph{feedback edge number} of a connected graph $G=(V,E)$, which is the minimum cardinality of a set $F\subseteq E$ of edges (called the \emph{feedback edge set}) such that $G-F$ is acyclic. The feedback edge number can be computed in quadratic time and has primarily been used to obtain fixed-parameter algorithms and polynomial kernels for problems where other parameterizations failed~\cite{GanianO21,BetzlerBNU12,BertertHHKN20,BevernFT20}.

Up to now, these were the only two edge-cut based graph parameters that have been considered in the broader context of algorithm design. This situation could be seen as rather unstisfactory in view of the large gap between the complexity of the richer class of graphs of bounded treecut width, and the significantly simpler class of graphs of bounded feedback edge number---for instance, the latter class is not even closed under disjoint union. Here, we propose a new parameter that lies ``between'' the feedback edge number and treecut width, and which can be seen as a localized relaxation of the feedback edge number: instead of measuring the total size of the feedback edge set, it only measures how many feedback edges can ``locally interfere with'' any particular part of the graph. 

Formally, for a connected graph $G=(V,E)$ and a spanning tree $T$ of $G$, let the \emph{local feedback edge set} at $v\in V$ be 
\ifshort
$E_{\loc}^T(v)=\{uw\in E\setminus E(T)~|~\text{ the unique path between }u\text{ and }w\text{ in }T\text{ contains }v\}.$
\fi
\iflong
\[E_{\loc}^T(v)=\{uw\in E\setminus E(T)~|~\text{ the unique path between }u\text{ and }w\text{ in }T\text{ contains }v\}.\]
\fi
The \emph{local feedback edge number of} $(G,T)$ (denoted $\lfes(G,T)$) is then equal to $\max_{v\in V} |E_{\loc}^T(v)|$, and the \emph{local feedback edge number of} $G$ is simply the smallest local feedback edge number among all possible spanning trees of $G$, i.e., $\lfes(G)=\min_{T\text{ is a spanning tree of }G} \lfes(G,T)$.

It is not difficult to show that the local feedback edge number is ``sandwiched'' between the feedback edge number and treecut width. We also show that computing it is \FPT.

\iflong
\begin{proposition}
\fi
\ifshort
\begin{proposition}
\fi
\label{pro:lfescompare}
For every graph $G$, $\tcw(G)\leq \lfes(G)+1$ and $\lfes(G)\leq \fes(G)$.
\end{proposition}

\iflong
\begin{proof}
Let us begin with the second inequality. Consider an arbitrary spanning tree $T$ of $G$. 
  Then for every $v \in V(G)$, $E_{\loc}^T(v)$ is a subset of a feedback edge set corresponding to the spanning tree $T$, so $|E_{\loc}^T(v)| \leq \fes(G)$ and the claim follows.

To establish the first inequality, we will use the notation and definition of treecut width from previous work~\cite[Subsection 2.4]{GanianKO21}.
 Let $T$ be the spanning tree of $G$ with $\lfes(G,T)=\lfes(G)$. We construct a treecut decomposition $(T, \XXX)$ where each bag contains precisely one vertex, notably by setting $X_t=\{t\}$ for each $t\in V(T)$.  Fix any node $t$ in $T$ other than root, let $u$ be the parent of $t$ in $T$. All the edges in $G\setminus ut$ with one endpoint in the rooted subtree $T_t$ and another outside of $T_t$ belong to $E^T_{loc}(t)$, so $\adh_T(t)=|\cut(t)|\leq |E^T_{loc}(t)| \leq \lfes(G)$.\\\\ Let $H_t$ be the torso of $(T, \XXX)$ in $t$, then $V(H_t)=\{t,z_1...z_l\}$ where $z_i$ correspond to connected components of $T \setminus t$, $i\in [l]$. In $\tilde H(t)$, only $z_i$ with degree at least $3$ are preserved. 
But all such $z_i$ are the endpoints of at least 2 edges in $|E^T_{loc}(t)|$, so $\tor(t)=|V(\tilde H_t)|\le 1+ |E^T_{loc}(t)| \le 1+ \lfes(G)$. Thus  $\tcw(G)\leq \lfes(G)+1$.
\end{proof}
\fi

\iflong
\begin{theorem}
\fi
\ifshort
\begin{theorem}
\fi
\label{thm:complfes}
The problem of determining whether $\lfes(G)\leq k$ for an input graph $G$ parameterized by an integer $k$ is fixed-parameter tractable. Moreover, if the answer is positive, we may also output a spanning tree $T$ such that $\lfes(G,T)\leq k$ as a witness.
 \end{theorem}

\iflong
\begin{proof}
Observe that since $\tcw(G)\leq \lfes(G)+1$ by Proposition~\ref{pro:lfescompare} and $\tw(G)\leq 2\tcw(G)^2+3\tcw(G)$~\cite{Ganian0S15}, we immediately see that no graph of treewidth greater than $k'=2k^2+5k+3$ can have a local feedback edge set of at most $k$. Hence, let us begin by checking that $\tw(G)\leq k'$ using the classical fixed-parameter algorithm for computing treewidth~\cite{Bodlaender96}; if not, we can safely reject the instance. 

Next, we use the fact that $\tw(G)\leq k'$ to invoke Courcelle's Theorem~\cite{Courcelle90,DowneyFellows13}, which provides a fixed-parameter algorithm for model-checking any \emph{Monadic Second-Order Logic} formula on $G$ when parameterized by the size of the formula and the treewidth of $G$. We refer interested readers to the appropriate books~\cite{CourcelleEngelfriet12,DowneyFellows13} for a definition of Monadic Second Order Logic; intuitively, the logic allows one to make statements about graphs using variables for vertices and edges as well as their sets, standard logical connectives, set inclusions, and atoms that check whether an edge is incident to a vertex. If the formula contains a free set variable $X$ and admits a model on $G$, Courcelle's Theorem allows us to also output an interpretation of $X$ on $G$ that satisfies the formula.

The formula $\phi$ we will use to check whether $\lfes(G)\leq k$ will be constructed as follows. $\phi$ contains a single free edge set variable $X$ (which will correspond to the sought-after feedback edge set). $\phi$ then consists of a conjunction of two parts, where the first part simply ensures that $X$ is a minimal feedback edge set using a well-known folklore construction~\cite{LangerRRS14,Barthesis}; this also ensures that $G-X$ is a spanning tree. In the second part, $\phi$ quantifies over all vertices in $G$, and for each such vertex $v$ it says there exist edges $e_1,\dots,e_k$ in $X$ such that for every edge $ab\in X$ distinct from all of $e_1,\dots, e_k$, there exists a path $P$ between $a$ and $b$ in $G-X$ which is disjoint from $v$. (Note that since the path $P$ is unique in $G-X$, one could also quantify $P$ universally and achieve the same result.) 

It is easy to verify that $\phi(X)$ is satisfied in $G$ if and only if $\lfes(G,G-X)\leq k$, and so the proof follows. Finally, we remark that---as with every algorithmic result arising from Courcelle's Theorem---one could also use the formula as a template to build an explicit dynamic programming algorithm that proceeds along a tree-decomposition of $G$.
\end{proof}
\fi

\section{Solving \BNSLneq\ with Parameters Based on Edge Cuts.}
\label{sec:edgeparams}
In this section we provide tractability and lower-bound results for \BNSLneq\ from the viewpoint of superstructure parameters based on edge cuts. Together with the previous lower bound that rules out fixed-parameter algorithms based on all vertex-separator parameters~\cite[Theorem 3]{OrdyniakS13}, the results presented here provide a comprehensive picture of the complexity of \BNSLneq\ with respect to superstructure parameterizations.

\iflong
\subsection{Using the Feedback Edge Number for \BNSLneq}
\label{sub:fes}
\fi
\ifshort
\noindent \textbf{Using the Feedback Edge Number for \BNSLneq.}\quad
\fi
We say that two instances $\III$, $\III'$ of \BNSL\ are \emph{equivalent} if (1) they are either both \YES-instances or both \NO-instances, and furthermore (2) a solution to one instance can be transformed into a solution to the other instance in polynomial time. Our aim here is to prove the following theorem:

\begin{theorem}
\label{thm:kernel}
There is an algorithm which takes as input an instance $\III$ of \BNSLneq\ whose superstructure has \fes\ $k$, runs in time $\bigoh(|\III|^2)$, and outputs an equivalent instance $\III'=(V',\FFF',\ell')$ of \BNSLneq\ such that $|V'|\leq 16k$.
\end{theorem}

In parameterized complexity theory, such data reduction algorithms with performance guarantees are called \emph{kernelization algorithms}~\cite{DowneyFellows13,CyganFKLMPPS15}. These may be applied as a polynomial-time preprocessing step before, e.g., more computationally expensive methods are used. The fixed-parameter tractability of \BNSLneq\ when parameterized by the \fes\ of the superstructure follows as an immediate corollary of Theorem~\ref{thm:kernel} (one may solve $\III$ by, e.g., exhaustively looping over all possible DAGs on $V'$ via a brute-force procedure). We also note that even though the number of variables of the output instance is polynomial in the parameter $k$, the instance $\III'$ need not have size polynomial in $k$.

We begin our path towards a proof of Theorem~\ref{thm:kernel} by computing a feedback edge set $E_F$ of $G$ of size $k$ in time $\bigoh(|\III|^2)$ by, e.g., Prim's algorithm. Let $T$ be the spanning tree of $G$, $E_F=E(G) \setminus E(T)$. The algorithm will proceed by the recursive application of certain reduction rules, which are polynomial-time operations that alter (``simplify'') the input instance in a certain way. A reduction rule is \emph{safe} if it outputs an instance which is equivalent to the input instance. We start by describing a rule that will be used to prune $T$ until all leaves are incident to at least one edge in $E_F$.

\begin{redrule}
\label{redruleone}
Let $v\in V$ be a vertex and let $Q$ be the set of neighbors of $v$ with degree $1$ in $G$.
 We construct a new instance $\III'=(V',\FFF',\ell)$ 
 by setting: 
 \textbf{\textup{1.}} $V':=V\setminus Q$; 
 \textbf{\textup{2.}} $\Gamma_{f'}(v):=\{\emptyset\}\cup \SB (P\setminus Q) \SM P\in \Gamma_f(v)\SE$;
  \textbf{\textup{3.}} for all $w\in V'\setminus \{v\}$, $f'_w=f_w$;
 \textbf{\textup{4.}} for every $P'\in \Gamma_{f'}(v)$: \\
$$f'_v(P'):=\max_{ P: P\setminus Q = P' } \big(f_v(P)+ \sum_{v_{\emph{in}}\in P \cap Q}f_{v_{\emph{in}}}(\emptyset) + \sum_{v_{\emph{out}}\in Q\setminus P}\max(f_{v_{\emph{out}}}(\emptyset),f_{v_{\emph{out}}}(v))\big).$$
          \end{redrule}

\iflong
\begin{lemma}
\fi
\ifshort
\begin{lemma}
\fi
\label{lem:redruleone}
Reduction Rule~\ref{redruleone} is safe.
\end{lemma}

\iflong
\begin{proof}
 For the forward direction, assume that $\III'$ admits a solution $D'$, and let $\lambda$ be the score $D'$ achieves on $v$. By the construction of $\III'$, there must be a parent set $Z\in \Gamma_f(v)$ such that $Z\cap V'=P_{D'}(v)$ (i.e., $Z$ agrees with $v$'s parents in $D'$) and $\lambda$ is the sum of the following scores: (1) $f_v(Z)$, (2) the maximum achievable score for each vertex in $Q\setminus Z$, and (3) the score of $\{\emptyset\}$ for each vertex in $Z\cap Q$. Let $D$ be obtained from $D'$ by adding the following arcs: $zv$ for each $z\in Z$, and $vq$ for each $q\in Q\setminus Z$ such that $q$ achieves its maximum score with $v$ as its parent. By construction, $\lambda=\sum_{w'\in \{v\}\cup Q}f_w(P_D(w))$. Since the scores of $D$ and $D'$ coincide on all vertices outside of $\{v\}\cup Q$ and $D$, we conclude that $\score(D)=\score(D')$, and hence $\III$ is a \YES-instance.
 
 For the converse direction, assume that $\III$ admits a solution $D$. Let $D'=D-Q$. By the construction of $f'_{v}$, it follows that $f'_v(P_{D'}(v))$ is greater or equal to the score $D$ achieves on $\{v\}\cup Q$. Thus, $D'$ is a solution to $\III'$, and we conclude that Reduction Rule~\ref{redruleone} is safe.
\end{proof}
\fi

Observe that the superstructure graph $G'$ obtained after applying one step of Reduction Rule~\ref{redruleone} is simply $G-Q$; after its exhaustive application we obtain an instance $\III$ such that all the leaves of the tree $T$ are endpoints of $E_F$.
Our next step is to get rid of long paths in $G$ whose internal vertices have degree $2$. 
We note that this step is more complicated than in typical kernelization results using feedback edge set as the parameter, since a directed path $Q$ in $G$ can serve multiple ``roles'' in a hypothetical solution $D$ and our reduction gadget needs to account for all of these. Intuitively, $Q$ may or may not appear as a directed path in $D$ (which impacts what other arcs can be used in $D$ due to acyclicity), and in addition the total score achieved by $D$ on the internal vertices of $Q$ needs to be preserved while taking into account whether the endpoints of $Q$ have a neighbor in the path or not. Because of this (and unlike in many other kernelization results of this kind~\cite{GanianO21,UhlmannW13,GanianKO21}), we will not be replacing $Q$ merely by a shorter path, but by a more involved gadget. 
\ifshort An illustration is provided in Figure~\ref{fig:redruletwo}.

\begin{redrule}
\label{redruletwo}
 Let $a,b_1,\dots,b_m,c$ be a path in $G$ such that for each $i\in [m]$, $b_i$ has degree precisely $2$, and let $P=\{b_1,\dots,b_m\}$. We construct a new instance $\III'=(V',\FFF',\ell)$ as follows:
 \begin{enumerate}
  \item[\textbf{\textup{1.}}] $V':=V\cup\{b\}\setminus \{b_2...b_{m-1}\}$;
\item[\textbf{\textup{2.}}] $\Gamma_{f'}(b)=\{B\cup \{b_1,b_m\}|B\subseteq \{a,c\}\}$ where $f'_b( B\cup\{b_1,b_m\})$ is equal to the maximum score that can be achieved by $P$ if $B$ are used as parents;
\item[\textbf{\textup{3.}}] The scores for $a$ and $c$ are obtained from $\FFF$ by simply adding $b$ to any parent set containing either $b_1$ or $b_m$;
\item[\textbf{\textup{4.}}] $\Gamma_{f'}(b_1)$ contains only $\{a,b,b_m\}$ with score equal to the maximum score that can be achieved by $P$ if $a$ is used as a parent but there is no path from $a$ to $b_m$ via $P$;
\item[\textbf{\textup{5.}}] $\Gamma_{f'}(b_m)$ contains only $\{c,b,b_1\}$ with score equal to the maximum score that can be achieved by $P$ if $c$ is used as a parent but there is no path from $c$ to $b_1$ via $P$;
\item[\textbf{\textup{6.}}] for all $w\in V' \setminus \{a,b_1, b,b_m,c\}$, $f'_w=f_w$.
\end{enumerate}
\end{redrule}
\fi

\iflong
\begin{redrule}
\label{redruletwo}
 Let $a,b_1,\dots,b_m,c$ be a path in $G$ such that for each $i\in [m]$, $b_i$ has degree precisely $2$. For each $B\subseteq \{a,c\}$, let $\ell_{\emph{max}}(B)$ be the maximum sum of scores that can be achieved by $b_1,\dots,b_m$ under the condition that $b_1$ (and analogously $b_m$) takes $a$ ($c$) into its parent set if and only if $a \in B$ ($c\in B$). In other words, $\ell_{\emph{max}}(B)=\max_{D_B}\sum_{b_{i}|i\in [m]}f_{b_i}(P_{D_B}(b_i))$ where $D_B$ is a DAG on $\{b_1,\dots,b_m\}\cup B$ such that $B$ does not contain any vertices of out-degree $0$ in $D_B$. Moreover, let $\ell_{\emph{noPath}}(a)$ (and analogously $\ell_{\emph{noPath}}(c)$) be the maximum score that can be achieved on the vertices $b_1,\dots,b_m$ by a DAG on $a,b_1,\dots,b_m,c$ with the following properties: $a$ ($c$) has out-degree $1$, $c$ ($a$) has out-degree $0$, and there is no directed path from $a$ to $b_m$ (from $c$ to $b_1$).

We construct a new instance $\III'=(V',\FFF',\ell)$ as follows:
\begin{itemize}
\item $V':=V\cup\{b\}\setminus \{b_2...b_{m-1}\}$;
\item $\Gamma_{f'}(b)=\{B\cup \{b_1,b_m\}|B\subseteq \{a,c\}\}$ with scores $f'_b( B\cup\{b_1,b_m\}):=\ell_\emph{max}(B)$;
\item The scores for $a$ and $c$ are obtained from $\FFF$ by simply adding $b$ to any parent set containing either $b_1$ or $b_m$; formally:
\begin{itemize}
\item $\Gamma_{f'}(a)$ is a union of $\{P \in \Gamma_{f}(a)|b_1\not \in P\}$, where $f'_a(P):=f_a(P)$ and  $\{P\cup \{b\}| b_1 \in P, P \in \Gamma_{f}(a)\}$, where $f'_a(P\cup \{b\}):=f_a(P)$;
\item $\Gamma_{f'}(c)$ is a union of $\{P \in \Gamma_{f}(c)|b_m\not \in P\}$, where $f'_c(P):=f_c(P)$, and  $\{P\cup \{b\}| b_m \in P, P \in \Gamma_{f}(c)\}$, where $f'_c(P\cup \{b\}):=f_c(P)$.
\end{itemize}
\item $\Gamma_{f'}(b_1)$ contains only $\{a,b,b_m\}$ with score $\ell_{\emph{noPath}}(a)$;
\item $\Gamma_{f'}(b_m)$ contains only $\{c,b,b_1\}$ with score $\ell_{\emph{noPath}}(c)$;
\item for all $w\in V' \setminus \{a,b_1, b,b_m,c\}$, $f'_w=f_w$.
\end{itemize}
\end{redrule}

An Illustration of Reduction Rule~\ref{redruletwo} is provided in Figure~\ref{fig:redruletwo}. The rule can be applied in linear time, since  the $6$ values of $\ell_{\emph{noPath}}$ and $\ell_{\emph{max}}$ can be computed in linear time by a simple dynamic programming subroutine that proceeds along the path $a,b_1,\dots,b_m,c$ (alternatively, one may instead invoke the fact that paths have treewidth $1$~\cite{OrdyniakS13}).
\fi

\begin{figure}[thb]
\begin{minipage}[c]{0.65\textwidth}
\vspace{-0.1cm}
\includegraphics[width=\textwidth]{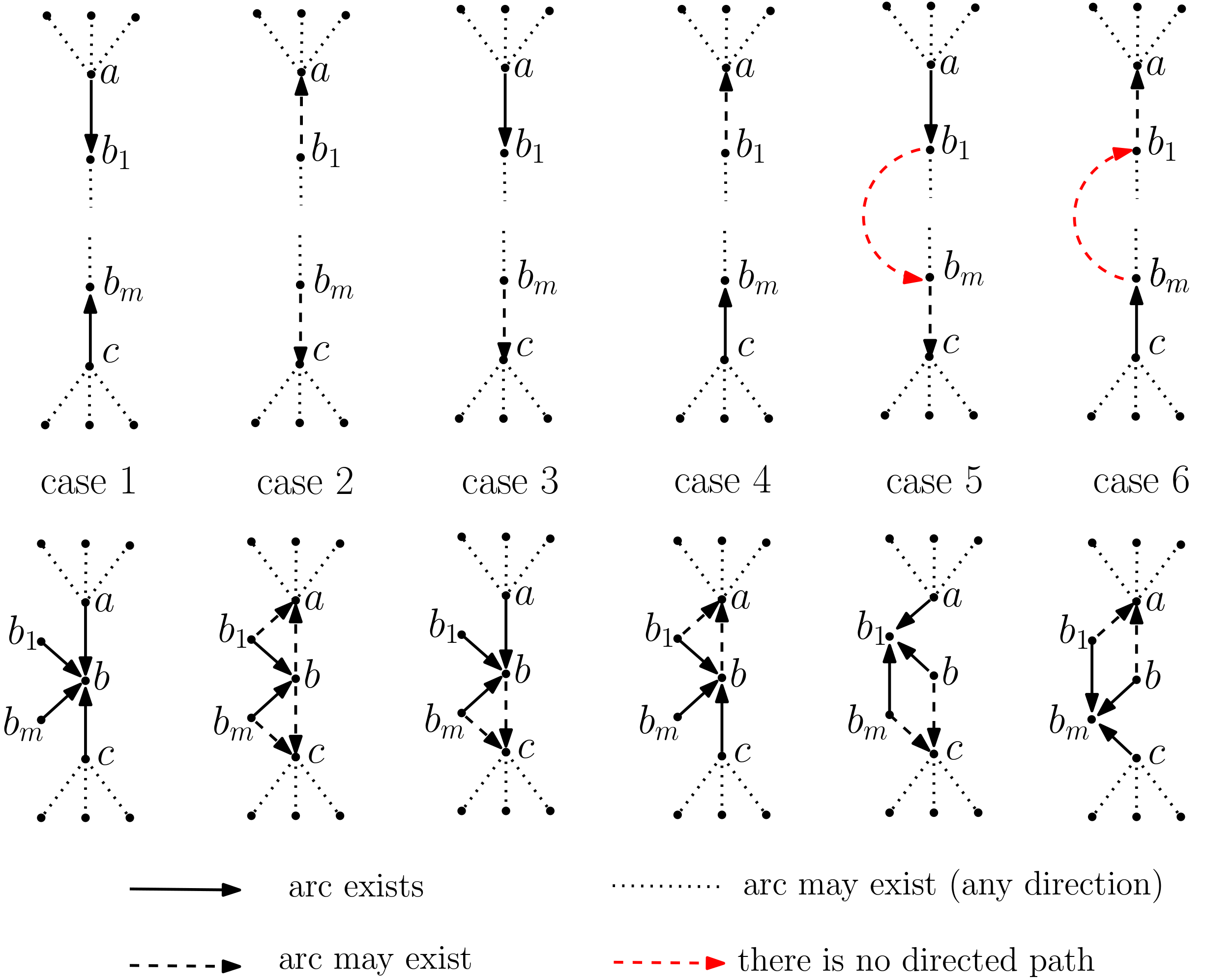}\vspace{-0.1cm}
\end{minipage}
\hfill
\begin{minipage}[c]{0.3\textwidth}
\caption{
\iflong ~\\ Top: The six possible scenarios that give rise to the values of $\ell_{\emph{max}}$ (Cases 1-4) and $\ell_{\emph{noPath}}$ (Cases 5-6). 
\fi
\ifshort ~\\
Top: The six possible solution scenarios that may arise when dealing with long paths. 
\fi \\
Bottom: The corresponding arcs in the gadget after the application of Reduction Rule~\ref{redruletwo}.
}
\label{fig:redruletwo}
\end{minipage}
\end{figure}

\iflong
\begin{lemma}
\fi
\ifshort
\begin{lemma}
\fi
\label{lem:redruletwo}
Reduction Rule~\ref{redruletwo} is safe.
\end{lemma}

\iflong
\begin{proof}
Note that the superstructure graph of reduced instance  is obtained from $G_{\III}$ by contracting $b_2...b_{m-1}$, adding $b$ and connecting it by edges to $a,c,b_1,b_m$.
We will show that a score of at least $\ell$ can be achieved in the original instance $\III$ if and only if a score of at least $\ell$ can be achieved in the reduced instance $\III'$.

Assume that $D$ is a DAG that achieves a score of $\ell$ in $\III$. We will construct a DAG $D'$, called the \emph{reduct} of $D$, with $f'(D')\geq \ell$. To this end, we first modify $D$ by removing the vertices $b_2...b_{m-1}$ and adding $b$ (let us denote the DAG obtained at this point $D^*$). Further modifications of $D^*$ depend only on $D[a,b_1...b_m,c]$, and we distinguish the $6$ cases listed below (see also Figure~\ref{fig:redruletwo}):
\begin{itemize}
    \item case 1: $D$ contains both arcs $ab_1$ and $cb_m$. We add to $D^*$ arcs from $a,c,b_1,b_m$ to $b$, denote resulting graph by $D'$. As $D'$ is obtained from DAG by making $b$ a sink, it is a DAG as well. Parent set of $b$ in $D'$ is $\{a,c,b_1,b_m\}$, so its score is $\ell_\emph{max}(a,c)\ge \sum_{i=1}^m f_{b_i}(P_D(b_i))$, which means that it achieves the highest scores all of $b_i$'s can achieve in $D$. The remaining vertices in $V(D')\setminus\{b_1,b_m,b\}$ have the same scores as in $D$, so $f'(D')\ge f(D) =\ell$.
    \item case 2: $D$ contains none of the arcs $ab_1$ and $cb_m$. To keep the scores of $a$ and $c$ the same as in $D$, we add to $D^*$ the arc $ba$ iff $D$ contains $b_1a$, add arc $bc$ iff $D$ contains $b_mc$. Furthermore, we add arcs $b_1 b$ and  $b_mb$ and denote resulting graph $D'$. As $D'$ is obtained from $D$ by making $b$ a source and then adding sources $b_1$ and $b_m$, it is a DAG as well. The parent set of $b$ in $D'$ is $\{b_1,b_m\}$, so its score is $\ell_\emph{max}(\emptyset)\ge \sum_{i=1}^m f_{b_i}( P_D(b_i))$. Rest of vertices in $V(D')\setminus\{b_1,b_m,b\}$ have the same scores as in $D$, so $f'(D')\ge f(D) =\ell$.
    \item case 3: $D$ doesn't contain arc $ab_1$, but contains $cb_m$ and all the arcs $b_{i+1} b_i$, $i\in [m-1]$.  We add to $D^*$ arcs $cb$, $b_1 b$ and  $b_m b$. We also add $ba$ iff $D$ contains $b_1 a$, to preserve the score of $a$. Denote resulting graph by $D'$. $D'$ can be considered as $D$ where long directed path $c\to b_m\to...\to b_1$ was replaced by $c\to b$ and then sources $b_1$ and $b_m$ were added, so it is a DAG. Arguments for scores are similar to cases 1 and 2.
    \item case 4: $D$ doesn't contain arc $cb_m$, but contains $a b_1$ and all the arcs $b_i b_{i+1}$, $i\in [m-1]$. This case is symmetric to case 3.
    \item case 5: $D$ contains the arc $a b_1$ but does not contain the arc $cb_m$ and at least one of the arcs $b_i b_{i+1}$, $i\in [m-1]$ is also missing (i.e., there is no directed path from $a$ to $b_m$).
We add to $D'$ arcs $bb_1$ and $b_m b_1$. If $b_m c\in A(D)$, add also $bc$. Denote the resulting graph $D'$. As $D'$ is obtained from $D^*$ by making $b_1$ a sink and $b$ a source, it is a DAG. $b_1$ has parent set $\{a,b,b_m\}$ in $D'$, so its score is $\ell_{\emph{noPath}}(a)\ge \sum_{i=1}^m f_{b_i} (P_D(b_i))$. Rest of vertices in $V(D')\setminus\{b_1,b_m,b\}$ have the same scores as in $D$, so $f'(D')\ge f(D) =\ell$.
    \item case 6: $D$ contains the arc $cb_m$ but does not contain the arc $ab_1$ and at least one of the arcs $b_{i+1} b_i$, $i\in [m-1]$ is also missing. This case is symmetric to case 5.
\end{itemize}

The considered cases exhaustively partition all possible configurations of $D[a,b_1...b_m,c]$, so we always can construct $D'$ with a score at least $\ell$.
For the converse direction, note that the DAGs constructed in cases 1-6 cover all optimal configurations on $\{a,b_1,b,b_m,c\}$: if there is a DAG $D''$ in $\III'$ with a score of $\ell'$, we can always reverse the construction to obtain a DAG $D'$ with score at least $\ell'$ such that  $D'[a,b_1,b,b_m,c]$  has one of the forms depicted at the bottom line of the figure. The claim for the converse direction follows from the fact that every such $D'$ is a reduct of some DAG $D$ of the original instance with the same score. 
\end{proof}

We are now ready to prove the desired result.
\fi

\begin{proof}[Proof of Theorem~\ref{thm:kernel}]  
We begin by exhaustively applying Reduction Rule~\ref{redruleone} on an instance whose superstructure graph has a feedback edge set of size $k$, which results in an instance with the same feedback edge set but whose spanning tree $T$ has at most $2k$ leaves. It follows that there are at most $2k$ vertices with a degree greater than $2$ in $T$.

Let us now ``mark'' all the vertices that either are endpoints of the edges in $E_F$ or have a degree greater then $2$ in $T$; the total number of marked vertices is upper-bounded by $4k$. We now proceed to the exhaustive application of Reduction Rule~\ref{redruletwo}, which will only be triggered for sufficiently long paths in $T$ that connect two marked vertices but contain no marked vertices on its internal vertices; there are at most $4k$ such paths due to the tree structure of $T$. Reduction Rule~\ref{redruletwo} will replace each such path with a set of $3$ vertices, and therefore after its exhaustive application we obtain an equivalent instance with at most $4k+4k\cdot3=16k$ vertices, as desired. Correctness follows from the safeness of Reduction Rules~\ref{redruleone}, \ref{redruletwo}, and the runtime bound follows by observing that the total number of applications of each rule as well as the runtime of each rule are upper-bounded by a linear function of the input size.
\end{proof}

As an immediate corollary of Theorem~\ref{thm:kernel}, we can apply a standard brute-force branching procedure~\cite{Ott2004} to solve \BNSL\ in time $n^{\bigoh(1)} + 2^{\bigoh(k)}$.

\iflong
\subsection{Fixed-Parameter Tractability of \BNSLneq\ using the Local Feedback Edge Number}
\label{sub:lfes}
\fi
\ifshort
\noindent \textbf{Fixed-Parameter Tractability of \BNSLneq\ using the Local Feedback Edge Number.} \quad
\fi
Our aim here will be to lift the fixed-parameter tractability of \BNSLneq\ established by Theorem~\ref{thm:kernel} by relaxing the parameterization to \lfes. In particular, we will prove:

\begin{theorem}
\label{thm:lfes} \BNSLneq\ is fixed-parameter tractable when parameterized by the local feedback edge number of the superstructure.
\end{theorem}

Since \fes\ is a more restrictive parameter than \lfes, this results in a strictly larger class of instances being identified as tractable. However, the means we will use to establish Theorem~\ref{thm:lfes} will be fundamentally different: we will not use a polynomial-time data reduction algorithm as the one provided in Theorem~\ref{thm:kernel}, but instead apply a dynamic programming approach. Since the kernels constructed by Theorem~\ref{thm:kernel} contain only polynomially-many variables w.r.t.\ \fes, that result is incomparable to Theorem~\ref{thm:lfes}.
\iflong

\fi
In fact one can use standard techniques to prove that, under well-established complexity assumptions, a data reduction result such as the one provided in Theorem~\ref{thm:kernel} \emph{cannot} exist for \lfes. \ifshort 
\begin{theorem}
\label{thm:nokernel}
Unless $\NP\subseteq \coNP/\cc{poly}$, there is no polynomial-time algorithm which takes as input an instance $\III$ of \BNSLneq\ whose superstructure has \lfes\ $k$ and outputs an equivalent instance $\III'=(V',\FFF',\ell')$ of \BNSLneq\ such that $|V'|\in k^{\bigoh(1)}$. In particular, \BNSLneq\ does not admit a polynomial kernel when parameterized by \lfes.
\end{theorem}
\fi \iflong
The intuitive reason for this is that \lfes\ is a ``local'' parameter that does not increase by, e.g., performing a disjoint union of two distinct instances (the same property is shared by many other well-known parameters such as treewidth, pathwidth, treedepth, clique-width, and treecut width).
We provide a formal proof of this claim at the end of Subsection~\ref{sub:lowerbounds}.

As our first step towards proving Theorem~\ref{thm:lfes}, we provide general conditions for when the union of two DAGs is a DAG as well. Let $D=(V,A)$ be a directed graph and $V'\subseteq V$. Denote by $\Con(V',D)$ the binary relation on $V'\times V'$ which specifies whether vertices from $V'$ are connected by a path in $D$: $\Con(V',D)=\{(v_1,v_2)\subseteq V'\times V'|\:\exists$ directed path from $v_1$ to $v_2$ in $D \}$. Similarly to arcs, we will use $v_1v_2\in $ as shorthand for $(v_1,v_2)$; we will also use $\trcl$ to denote the transitive closure.

\begin{lemma}\label{lem:con2dags}
Let $D_1$, $D_2$ be directed graphs with common vertices $V_{\emph{com}}=V(D_1)\cap V(D_2)$, $V_{\emph{com}} \subseteq V_1 \subseteq V(D_1)$, $V_{\emph{com}} \subseteq V_2\subseteq V(D_2)$. Then:
 \begin{itemize}
     \item (i) $\Con(V_1\cup V_2, D_1\cup D_2)=\trcl(\Con(V_1,D_1)\cup \Con(V_2,D_2))$;
     \item (ii) If $D_1$, $D_2$ are DAGs and $\Con(V_1\cup V_2, D_1\cup D_2$) is irreflexive, then $D_1\cup D_2$ is a DAG. 
 \end{itemize}
\end{lemma}

\begin{proof}
(i) Denote $R_i:=\Con(V_i,D_i)$, $i=1,2$. Obviously $\trcl(R_1\cup R_2)$ is a subset of $\Con(V_1\cup V_2, D_1\cup D_2)$. Assume that for some $x$, $y \in V_1\cup V_2$ there exists a directed path $P$ from $x$ to $y$ in $D_1\cup D_2$. We will show (by induction on the length $l$ of shortest $P$) that $xy \in \trcl(R_1\cup R_2)$.
\begin{itemize}
     \item $l=1$: in this case there is an arc $xy$ in some $D_i$, so $xy \in R_i \subseteq \trcl(R_1\cup R_2)$
     \item $l\to l+1$. If P is completely contained in some $D_i$, then $xy \in R_i \subseteq \trcl(R_1\cup R_2)$. Otherwise $P$ must contain arcs $e\notin A(D_1)$, $f\notin A(D_2)$. Then there is $w\in V_{\emph{com}} \subseteq V_1 \cup V_2$ between them. By the induction hypothesis $xw \in \trcl(R_1\cup R_2)$ and  $wy \in \trcl(R_1\cup R_2)$, so $xy \in \trcl(R_1\cup R_2)$
 \end{itemize}
 (ii) The precondition implies that the digraph $D_1\cup D_2$ induced on $V_1\cup V_2$ is a DAG. 
 Assume that $D_1\cup D_2$ is not a DAG and let $C$ be a shortest directed cycle in $D_1\cup D_2$. As $D_1$ and $D_2$ are DAGs, $C$ must contain arcs $e\notin A(D_1)$, $f\notin A(D_2)$. So there are least 2 different vertices $x,y$ from $V_{\emph{com}}$ in $C$. By (i) we have that $xy \in \trcl(R_1\cup R_2)$ and $yx \in \trcl(R_1\cup R_2)$, then also $xx \in \trcl(R_1\cup R_2)$, which contradicts irreflexivity.
\end{proof}

\fi

Towards proving Theorem~\ref{thm:lfes}, assume that we are given an instance  $\III=(V,\FFF,\ell)$ of \BNSLneq\ with connected superstructure graph $G=(V,E)$. Let $T$ be a fixed rooted spanning tree of $G$ such that $\lfes(G,T)=\lfes(G)=k$, denote the root by $r$.
For $v \in V(T)$, let $T_v$ be the subtree of $T$ rooted at $v$, let $V_v=V(T_v)$, and let $\bar V_v=N_G(V_v)\cup V_v$. We define the \emph{boundary} $\delta(v)$ of $v$ to be the set of endpoints of all edges in $G$ with precisely one endpoint in $V_v$ (observe that the boundary can never have a size of $1$). 
\ifshort
Notice that $|\delta(v)|\leq 2k+2$.
\fi
\iflong
$v$ is called \emph{closed} if $|\delta(v)| \le 2$ and \emph{open} otherwise. 
We begin by establishing some basic properties of the local feedback edge set.

\begin{observation}
\label{obs:basiclfesproperties}
Let $v$ be a vertex of $T$. Then:
\begin{enumerate}[itemsep=0pt]
\item For every closed child $w$ of $v$ in $T$, it holds that $\delta (w)=\{v,w\}$ and $vw$ is the only edge between $V_w$ and $V\setminus V_w$ in $G$.
\item $|\delta(v)|\leq 2k+2$.  \item Let $\{v_i|i\in [t]\}$ be the set of all open children of $v$ in $T$. Then $t\leq 2k$ and\\ $\delta(v)\subseteq \cup _{i=1}^t \delta(v_i) \cup \{v\}\cup N_G(v)$
\end{enumerate}
\end{observation}

\begin{proof}
The first claim follows by the connectivity assumption on $G$ and the definition of boundary. 

For the second claim, clearly $\delta(r)=\emptyset$. Let $v\neq r$ have the parent $u$, and consider an arbitrary $w\in \delta(v)\setminus \{u,v\}$. Then there is an edge $ww'\in E(G)$ with precisely one endpoint in $V_v$ and $ww'\ne uv$. Hence $ww' \not \in E(T)$ and the path between $w$ and $w'$ in T contains $v$, and this implies $ww'\in E_{loc}^T(v)$ by definition. Consequently, $w\in V_{loc}^T(v)$. 
For the claimed bound we note that $|V_{loc}^T(v)| \le 2 |E_{loc}^T(v)| \le 2k$.

For the third claim, let $w=v_i$ for some $i\in[t]$. As $w$ is open, there exists an edge $e\neq vw$ between $V_w$ and $V\setminus V_w$ in $G$. By definition of local feedback edge set, $e \in E_{loc}^T(v)$. Let $x_w$ be the endpoint of $e$ that belongs to $V_w$, then $x_w \in V_{loc}^T(v)$ and $x_w\not \in V_{w'}$ for any open child $w'\neq w$ of $v$. But $|V_{loc}^T(v)| \le 2k$, which yields the bound on number $t$ of open children. \\
For the boundary inclusion, consider any edge $c$ in $G$ with precisely one endpoint $x_v$ in $V_v$. Note that $x_v$ can not belong to $V_w$ for any closed child $w$ of $v$. If $x_v\in V_{v_i}$ for some $i\in [t]$, then endpoints of $c$ belong to $\delta(v_i)$. Otherwise $x_v=v$ and therefore the second endpoint of $c$ is in $N_G(v)$.
\end{proof}

With Observation~\ref{obs:basiclfesproperties} in hand, we can proceed to a definition of the records used in our dynamic program. 
\fi
\ifshort
We can now proceed to a definition of the records that will be used in our dynamic program.
\fi
Intuitively, these records will be computed in a leaf-to-root fashion and will store at each vertex $v$ information about the best score that can be achieved by a partial solution that intersects the subtree rooted at $v$.

Let $R$ be a binary relation on $\delta(v)$ and $s$ an integer. For $s\in \mathbb{Z}$,  we say that $(R:s)$ is a \emph{record} for a vertex $v$ if and only if there exists a DAG $D$ on $\bar V_v$ such that 
\ifshort
(1) $w\in V_v$ for each arc $uw\in A(D)$, (2) $ab\in R$ if there exists an $a$-$b$ path in $D$, and (3) $s$ is the total score achieved by $D$ on vertices in $V_v$.
\fi
\iflong
(1) $w\in V_v$ for each arc $uw\in A(D)$, (2) $R=\Con(\delta(v), D)$ and (3) $\sum_{u\in V_v} f_u(P_D(u)) = s$. \fi
The records $(R,s)$ where $s$ is maximal for fixed $R$ are called \emph{valid}. Denote the set of all valid records for $v$ by $\mathcal R (v)$, and note that $|\mathcal{R}(v)|\leq 2^{\bigoh(k^2)}$.

Observe that if $v_i$ is a closed child of $v$, then \iflong by Observation \ref {obs:basiclfesproperties}.1 \fi $\mathcal R(v_i)$ consists of precisely two valid records: one for $R=\emptyset$ and one for $R=\{v v_i\}$. Moreover, the root $r$ of $T$ has only a single valid record $(\emptyset:s_{\III})$, where $s_\III$ is the maximum score that can be achieved by a solution in $\III$. The following lemma lies at the heart of our result and shows how we can compute our records in a leaf-to-root fashion along $T$.

\iflong
\begin{lemma}
\fi
\ifshort
\begin{lemma}
\fi
\label{lem:combRecordsnew} 
Let $v\in V(G)$ have $m$ children in $T$ where $m>0$, and assume we have computed $\mathcal{R}(v_i)$ for each child $v_i$ of $v$. Then $\mathcal{R}(v)$ can be computed in time at most $m\cdot |\Gamma_f(v)|\cdot 2^{\bigoh(k^3)}$.
\end{lemma}

\iflong
\begin{proof}
Without loss of generality, let the open children of $v\in V(G)$ be $v_1,\dots,v_t$ and let the remaining (i.e., closed) children of $v$ be $v_{t+1},\dots,v_m$; recall that by Point 3. of Observation~\ref{obs:basiclfesproperties}, $t\leq 2k$. For each closed child $v_j$, $j\in [m]\setminus [t]$, let $s^{\emptyset}_j$ be the second component of the valid record for $\emptyset\in \mathcal{R}(v_j)$, and let $s^{\times}_j$ be the second component of the valid record for the single non-empty relation in $\mathcal{R}(v_j)$. Consider the following procedure $\mathbb{A}$. 

First, $\mathbb{A}$ branches over all choices of $P\in \Gamma_f(v)$ and all choices of $(R_i,s_i)\in \mathcal{R}(v_i)$ for each individual open child $v_i$ of $v$. Let $R_0=\SB pv \SM p\in P \SE$ and let $R'=\bigcup_{j\in [t]_0} R_j$.  If $\trcl(R')$ is not irreflexive, we discard this branch; otherwise, we proceed as follows. Let $R_\emph{new}$ be the subset of $R'$ containing all arcs $uw$ such that $w\in V_v$. Moreover, let $s_\emph{new}=f_v(P)+(\sum_{i\in [t]}s_i)+(\sum_{i\in [m]\setminus [t]~|~v_i\in P}s^\emptyset_i)+(\sum_{i\in [m]\setminus [t]~|~v_i\not \in P}(\max (s^\emptyset_i,s^\times_i))$. 

The algorithm $\mathbb{A}$ gradually constructs a set $\mathcal{R}^*(v)$ as follows. At the beginning, $\mathcal{R}^*(v)=\emptyset$. For each newly obtained tuple $(R_\emph{new},s_\emph{new})$, $\mathbb{A}$ checks whether $\mathcal{R}^*(v)$ already contains a tuple with $R_\emph{new}$ as its first element; if not, we add the new tuple to $\mathcal{R}^*(v)$. If there already exists such a tuple $(R_\emph{new},s_\emph{old})\in \mathcal{R}^*(v)$, we replace it with $(R_\emph{new},\max(s_\emph{old},s_\emph{new}))$.
 
For the running time, recall that in order to construct $\mathcal R^* (v)$ the algorithm branched over $|\Gamma_f(v)|$-many possible parent sets of $v$ and over the choice of at most $2k$-many binary relations $R_i$ on the boundaries of open children. According to Observation~\ref{obs:basiclfesproperties}.2, there are at most $3^{(2k+2)^2}$ options for every such relation, so we have at most $\bigoh((3^{(2k+2)^2})^{2k}\cdot|\Gamma_f(v)|)  \le 2^{\bigoh(k^3)}\cdot |\Gamma_f(v)|$ branches. In every branch we compute $\trcl(R')$ in time $k^{\bigoh(1)}$ and then compute the value of $s_\emph{new}$ using the equation provided above before updating $\mathcal R^* (v)$, which takes time at most $\bigoh(m)$.
 
Finally, to establish correctness it suffices to prove following claim:
\begin{claim}
\label{cl:combRecords}
$(R:s)$ is a record for $v$ if and only if there exist $P\in \Gamma_f(v)$ and records $(R_i:s_i)$ for $v_i$, $i\in[m]$, such that:
\begin{itemize}
    \item $\trcl (\cup_{i=0}^t R_i)$ is irreflexive;
    \item $R_i=\emptyset$ for any closed child $v_i\in P$; 
    \item $\sum_{i=1}^m s_i+f_v(P)=s$;
    \item $R= (\trcl (\cup_{i=0}^t R_i))|_{\delta(v)\times \delta(v)}$.
\end{itemize}
Moreover, if $(R:s)\in \mathcal R(v)$ then in addition:
  \begin{itemize}
  \item $(R_i:s_i)\in \mathcal R(v_i)$, $i \in [t]$;
  \item for every closed child $v_i\not \in P$, $s_i=\max (s^\emptyset_i,s^\times_i)$.
\end{itemize}
\end{claim}
 \emph{Proof of the Claim.}
(a) ($\Leftarrow$) Denote $V_i=V_{v_i}$ and $\bar V_i=\bar V_{v_i}$, $i\in[m]$.  For every $i\in[m]$ there exists DAG $D_i$ on $\bar V_i$ such that all its arcs finish in $V_i$, $R_i=\Con(\delta(v_i), D_i)$ and $\sum_{u\in V_i} f_u(P_{D_i}(u)) = s_i$. Denote by $D_0$ DAG on $V_0=v\cup N_G(v)$ with arc set $R_0$. We will construct the witness $D$ of $(R,s)$ by gluing together all $D_i$, $i\in [m]_0$. \\\\
 We start from $D_0$ and DAGs of open children. Note that $\Con(V_0, D_0)=R_0$ and $\Con(\delta(v_i), D_i)=R_i$ for $i\in [t]$ . Inductive application of Lemma \ref{lem:con2dags} to DAGs $D_i$, $i\in[t]$, yields $\Con(\cup _{i=1}^t \delta(v_i) \cup V_0, D^*)=\trcl(\cup _{i=0}^t R_i)$. In particular, as $ \delta(v)\subseteq \cup _{i=1}^t \delta(v_i) \cup V_0$  by Observation~\ref{obs:basiclfesproperties}.3, we have that $\Con(\delta(v), D^*)= (\trcl(\cup _{i=0}^t R_i))|_{\delta(v)\times \delta(v)}=R$. As $\trcl (\cup _{i=0}^t R_i)$ is irreflexive, $D^*=\cup_{i=0}^t D_i$ is DAG by Lemma \ref{lem:con2dags}.\\\\Now we add to $D^*$ DAGs for closed children and finally obtain $D=\cup_{i=t+1}^m D_i \cup D^*$. For every closed child $v_i$, $D_i$ is by Observation~\ref{obs:basiclfesproperties}.1 the union of $v$ and  $D_i\setminus v$, plus at most one of arcs $vv_i, v_iv$ between them (recall $R_i=\emptyset$ for any closed child $v_i\in P$). Note that  $D_i\setminus v$ can share only $v_i$ with $D_0$ and doesn't have common vertices with any other $D_j$.  Therefore any directed path in $D$ starting and finishing outside outside of $V_i$, $i>t$, doesn't intersect $V_i$. In particular, acyclicity of $D^*$ and $D_i$, $i\in[m]\setminus[t]$, implies acyclicity of $D$; $\Con(\delta(v), D)=\Con(\delta(v), D^*)=R$.\\\\ All the arcs in $D_i$ finish in $V_i$, so parent set for every $x_i\in D_i$ in $D$ is the same as in $D_i$, $ i \in [m]$. Also parent set of $v$ in $D$ is the same as in $D_0$. So $$\sum_{u\in V_v} f_u(P_D(u)) = \sum_{i=1}^m \sum_{u\in V_i} f_u(P_{D_i}(u)) + f_v(P_{D_0}(v)) = \sum_{i=1}^m s_i+f_v(P)=s$$   

 $(\Rightarrow)$ Let $D$ be a witness for $(R:s)$, i.e. $D$ is DAG on $\bar V_v$ with all arcs finishing in $V_v$ such that $\sum_{u\in V_v} f_u(P_D(u)) =s$ and $\Con(\delta(v),D)=R$. For $i=1\in[m]$ define $D'_i=D[\bar V_i]$ and let $D_i$ be obtained from $D'_i$ by deleting arcs that finish outside $V_i$. Note that $\cup_{i=1}^m D_i= D$. Let $R_i=\Con(\delta(v_i),D_i)$, as in ($\Leftarrow$) we have that $R=\Con(\delta(v), D)=\trcl (\cup_{i=0}^t R_i))|_{\delta(v)\times \delta(v)}$. As $D$ is DAG, $\trcl (\cup_{i=0}^t R_i)$ is irreflexive and $R_i=\emptyset$ for any closed child $v_i\in P$.  Local score for $D_i$ is $$s_i=\sum_{u\in V_i} f_u(P_{D_i}(u))=\sum_{u\in V_i} f_u(P_{D'_i}(u))=\sum_{u\in V_i} f_u(P_{D}(u))$$ So $v_i$ has record $(R_i:s_i)$. Denote $P=P_D(v)$. Then:
 $$s=\sum_{u\in V_v} f_u(P_D(u)) = \sum_{i=1}^m \sum_{u\in V_i} f_u(P_D(u)) + f_v(P_D(v)) = \sum_{i=1}^m s_i+f_v(P)$$
 (b)
 Let $(R:s)\in \mathcal R(v)$ and all $D, P, D_i, R_i, s_i$ are as in $(a)(\Rightarrow)$. Assume that for some $i$ $(R_i,s_i)$ is not valid record of $v_i$. In this case $v_i$ must have a record $(R_i:s_i+\Delta)$ with $\Delta > 0$. But then $(a)(\Leftarrow)$ implies that $v$ has record $(R:s+\Delta)$, which contradicts to validity of $(R:s)$\\\\
Assume that some closed $v_i\not \in P$ has valid record $(R_i', s_i+\Delta)$ with $\Delta > 0$. $R'$ and $R$ differ only by arc $vv_i$, so addition or deletion of the arc to $D$ would increase the total score by $\Delta > 0$ without creating cycles. This would result in record $(R:s+\Delta)$ and yield a contradiction with validity of $(R:s)$. \hfill $\blacksquare$
 \end{proof}

We are now ready to prove the main result of this subsection.

\begin{proof}[Proof of Theorem~\ref{thm:lfes}] We provide an algorithm that solves \BNSLneq\ in time $2^{O(k^3)}\cdot n^3$, where $n=|\III|$, assuming that a spanning tree $T$ of $G$ such that $\lfes(G,T)=k$ is provided as part of the input. Once that is done, the theorem will follow from Theorem~\ref{thm:complfes}.

The algorithm computes $\mathcal{R} (v)$ for every node $v$ in $T$, moving from leaves to the root:
\begin{itemize}
    \item For a leaf $v$, compute $\mathcal{R}^*(v) := \{(R_P:f_v(P))|\:P\in \Gamma_f(v),\: R_P =  \{uv|u \in P\}  \}$. This can be done by simply looping over $\Gamma_f(v)$ in time $\bigoh(n)$.
         Note that $\mathcal{R}^*(v)$ is the set of all records of $v$, so we can correctly set  $\mathcal R(v):= \{(R:s)\in \mathcal{R}^* (v) |$ there is no $(R:s')\in \mathcal R^* (v)$ with $s'>s\}$.
    \item Let $v\in V(G)$ have at least one child in $T$, and assume we have computed $\mathcal{R}(v_i)$ for each child $v_i$ of $v$. Then we invoke Lemma~\ref{lem:combRecordsnew} to compute $\mathcal{R}(v)$  in time at most $m\cdot |\Gamma_f(v)|\cdot 2^{\bigoh(k^2)} \leq 2^{\bigoh(k^2)}\cdot n^2$.\qedhere
\end{itemize}
 \end{proof}
\fi

\ifshort
To prove Theorem~\ref{thm:lfes}, we start by invoking Theorem~\ref{thm:complfes} to obtain a spanning tree $T$ and then compute the records $\mathcal{R}(v)$ for each leaf of $T$ via exhaustive branching. We then apply Lemma~\ref{lem:combRecordsnew} to propagate our record sets towards the root $r$ of $T$; once we obtain $\mathcal{R}(r)$, we can output a solution in constant time. The runtime of the dynamic programming procedure used in the proof of Theorem~\ref{thm:lfes} is upper-bounded by $|\III|^3\cdot 2^{\bigoh(k^3)}$. 

For our final result for this section, recall that \lfes\ lies between \fes\ and treecut width in the parameter hierarchy (see Proposition~\ref{pro:lfescompare}). Since \BNSLneq\ is \FPT\ when parameterized by \lfes, the next step would be to ask whether this tractability result can be lifted to treecut width. Below, we answer this question negatively via a reduction from the \W$[1]$-hard \textsc{Multicolored Clique} problem~\cite{DowneyFellows13,CyganFKLMPPS15}.

\begin{theorem}[]
\label{thm:tcwhard}
\BNSLneq is \W$[1]$-hard when parameterized by the treecut width of the superstructure.
  \end{theorem}
\fi

\iflong
\subsection{Lower Bounds for \BNSLneq}
\label{sub:lowerbounds}
Since \lfes\ lies between \fes\ and treecut width in the parameter hierarchy (see Proposition~\ref{pro:lfescompare}) and \BNSLneq\ is \FPT\ when parameterized by \lfes, the next step would be to ask whether this tractability result can be lifted to treecut width. Below, we answer this question negatively.
     
\begin{theorem}
\label{thm:tcwhard}
\BNSLneq is \W$[1]$-hard when parameterized by the treecut width of the superstructure graph.
  \end{theorem}

In fact, we show an even stronger result: \BNSLneq\ is $\W[1]$-hard when parameterized by the vertex cover number of the superstructure even when all vertices outside of the vertex cover are required to have degree at most $2$. We remark that while \BNSLneq\ was already shown to be $\W[1]$-hard when parameterized by the vertex cover number~\cite{OrdyniakS13}, in that reduction the degree of the vertices outside of the vertex cover is not bounded by a constant and, in particular, the graphs obtained in that reduction have unbounded treecut width.

\begin{proof}[Proof of Theorem~\ref{thm:tcwhard}]
  
We reduce from the following well-known $\W[1]$-hard problem~\cite{DowneyFellows13,CyganFKLMPPS15}:

\noindent
\begin{center}
\begin{boxedminipage}{0.98 \columnwidth}
\textsc{Regular Multicolored Clique} ($\RMC$) \\[5pt]
\begin{tabular}{l p{0.83 \columnwidth}}
Input: & A $k$-partite graph $G=(V_1\cup...\cup V_k, E)$ such that $|N_G(v)|=m$ for every $v \in V$\\
Parameter: & The integer $k$\\
Question: &  Are there nodes $v^i$ that form a $k$-colored clique in $G$, i.e. $v^i\in V_i$ and $v^i v^j \in E$ for all $i,j\in[k],i \neq j$?
\end{tabular}
\end{boxedminipage}
\end{center}
We say that vertices in $V_i$ have color $i$.
Let $G=(V_1\cup...\cup V_k,E)$ be an instance of $\RMC$. We will construct an instance $(V, \FFF, \ell)$ of 
\BNSLneq such that $\III$ is a $\YES$-instance if and only if $G$ is a $\YES$-instance of $\RMC$.
$V$ consists of one vertex $v_i$ for each color $i \in [k]$ and one vertex $v_e$ for every edge $e \in E$. For each edge $e\in E$ that connects a vertex of color $i$ with a vertex of color $j$, the constructed vertex $v_e$ will have precisely one element in its score function that achieves a non-zero score, in particular: $f_{v_e}(\{v_i,v_j\})=1$.

Next, for each $i\in [k]$, we define the scores for $v_i$ as follows. For every $v\in V_i$, let $E_v$ be the set of all edges incident to $v$ in $G$, and let $P_i^v=\{v_e:e\in E_v\}$. We  now set $f_{v_i}(P_i^v)=m+1$ for each such $v$; all other parent sets will receive a score of $0$.
Note that $\SB v_i \SM i\in [k] \SE$ forms a vertex cover of the superstructure graph and that all vertices outside of this vertex cover have degree at most $2$, as desired.
We will show that $G$ has a $k$-colored clique if and only if there is a Bayesian network $D$ with score
at least $\ell=|E|+k+{k \choose 2}$. (In fact, it will later become apparent that the score can never exceed $\ell$.)

Assume first that $G$ has a $k$-colored clique on $v^i, i\in[k],$ consisting of a set $E_X$ of ${k\choose 2}$ edges. Consider the digraph $D$ on $V$ obtained as follows. For each vertex $v_i$, $i\in [k]$, and each vertex $v_e$ where $e\in E$, $D$ contains the arc $v_ev_i$ if $v_e$ is incident to $v^i$ and otherwise $D$ contains the arc $v_iv_e$. This completes the construction of $D$. Now notice that the construction guarantees that each $v_i$ receives
the parent set $P_i^{v^i}$ and hence contributes a score of $m+1$. Moreover, for every edge $e$ not incident to a vertex in the clique, the vertex $v_e$ contributes a score of $1$; note that the number of such edges is $|E|-km+{k \choose 2}$; indeed, every $v_i$ is incident to $m$ edges but since $v^i,i\in[k],$ was a clique we are guaranteed to double-count precisely ${k \choose 2}$ many edges.
Hence the total score is $k(m+1)+|E|-km+{k \choose 2}=|E|+k+{k \choose 2}$, as desired.

Assume that $\III=(V, \FFF, \ell)$ is a $\YES$-instance and let $s_{\opt}\geq \ell = |E|+k+{k \choose 2}$ be the maximum score that can be achieved by a solution to $\III$; let $D$ be a dag witnessing such a score. 
 Then all $v_i$, $i\in[k]$, must receive a score of $m+1$ in $D$. Indeed, assume that some $v_i$ receives a score of $0$ and let $P_v$ be any parent set of $v_i$ with a score $m+1$. Modify $D$ by 
orienting edges $v_i v_e$ for every $v_e\in P_v$ inside $v_i$. Now local score of $v_i$ is $m+1$, total score of the rest of vertices decreased by at most $m$ 
(maximal number of $v_e$ that had local score 1 in $D$ and lost it after the modification). So the modified DAG has a score of at least $s_{\opt}+1$, which 
contradicts the optimality of $s_{\opt}$. Therefore all $v_i$, $i\in[k]$, get score $m+1$ in $D$. 

Let $P_i$ be parent set of $v_i$ in $D$, then
$|P_i|=m$, $P_i=P_i^{v^i}$ for some $v^i \in V_i$. For every $v_e \in P_i$, the local score of $v_e$ in $D$ is $0$. Denote by $E_{\emph{unsat}}$ the set of all $v_e$ that
have a score of $0$ in $D$. Every $v_e$ belongs to at most 2 different $P_i$ 
and $P_i\cap P_j \le 1$ for every $i \neq j$, so $|E_{\emph{unsat}}|\ge km -{k \choose 2}$. If $|E_{\emph{unsat}}|> km -{k \choose 2}$, 
sum of local scores of $e_v$ in $D$ would be smaller then  $|E|-km+{k \choose 2}$, which results in $s_{\opt} < |E|+k+{k \choose 2}$.
Therefore $|E_{\emph{unsat}}|= km -{k \choose 2}$. But this means that $P_i\cap P_j \neq \emptyset$ for any  $i\neq j$, i.e. $v^i$, $i\in [k]$ form a $k$-colored clique in $G$. 
In particular $s_{\opt}=\ell.$
\end{proof}

For our second result, we note that the construction in the proof of Theorem~\ref{thm:tcwhard} immediately implies that \BNSLneq\ is \NP-hard even under the following two conditions: (1) $\ell+\sum_{v\in V}|\Gamma_f(v)|\in \bigoh(|V|^2)$ (i.e., the size of the parent set encoding is quadratic in the number of vertices),
 and (2) the instances are constructed in a way which makes it impossible to achieve a score higher than $\ell$. Using this, as a fairly standard application of \emph{AND-cross-compositions}~\cite{CyganFKLMPPS15} we can exclude the existence of an efficient data reduction algorithm for \BNSLneq\ parameterized by \lfes:

\begin{theorem}
\label{thm:nokernel}
Unless $\NP\subseteq \coNP/\cc{poly}$, there is no polynomial-time algorithm which takes as input an instance $\III$ of \BNSLneq\ whose superstructure has \lfes\ $k$ and outputs an equivalent instance $\III'=(V',\FFF',\ell')$ of \BNSLneq\ such that $|V'|\in k^{\bigoh(1)}$. In particular, \BNSLneq\ does not admit a polynomial kernel when parameterized by \lfes.
\end{theorem}

\begin{proof}[Proof Sketch]
We describe an AND-cross-composition for the problem while closely following the terminology and intuition introduced in Section 15 in the book~\cite{CyganFKLMPPS15}. Let the input consist of instances $\III_1,\dots,\III_t$ of (unparameterized) instances of \BNSLneq\ which satisfy conditions (1) and (2) mentioned above, and furthermore all have the same size and same target value of $\ell_1$ (which is ensured through the use of the polynomial equivalence relation $\mathcal{R}$~\cite[Definition 15.7]{CyganFKLMPPS15}). The instance $\III$ produced on the output is merely the disjoint union of instances $\III_1,\dots,\III_t$ where we set $\ell:=t\cdot \ell_1$, and we parameterize $\III$ by \lfes. 

Observe now that condition (a) in Definition 15.7~\cite{CyganFKLMPPS15} is satisfied by the fact that the local feedback edge number of $\III$ is upper-bounded by the number of edges in a connected component of $\III$. Moreover, the AND- variant of condition (b) in that same definition (see Subsection 15.1.3~\cite{CyganFKLMPPS15}) is satisfied as well: since none of the original instances can have a score greater than $\ell_1$, $\III$ achieves a score of $\ell_1\cdot t$ if and only if each of the original instances was a \YES-instance. 

This completes the construction of an AND-cross-composition for \BNSLneq\ parameterized by \lfes, and the claim follows by Theorem 15.12~\cite{CyganFKLMPPS15}.
\end{proof}
\fi

\newcommand{\fp}[1]{f_{#1}^{\neq 0}}
\newcommand{\fs}[1]{f_{#1}^{+}}

\section{Additive Scores and Treewidth}
\label{sec:add}
While the previous section focused on the complexity of \BNSL\ when the non-zero representation was used (i.e., \BNSLneq), here we turn our attention to the complexity of the problem with respect to the additive representation. Recall from Subsection~\ref{sec:prelims} that there are two variants of interest for this representation: \BNSLadd\ and \BNSLaddq. We begin by showing that, unsurprisingly, both of these are \NP-hard. \ifshort
We do so by reducing from the classical \textsc{Minimum Feedback Arc Set} problem~\cite{Gavril77b,DemetrescuF03}.
\fi

\iflong
\begin{theorem}
\fi
\ifshort
\begin{theorem}[]
\fi
\label{thm:addhard}
\BNSLadd\ is \NP-hard. Moreover, \BNSLaddq\ is \NP-hard for every $q\geq 3$.
\end{theorem}

\iflong
\begin{proof}
We provide a direct reduction from the following \NP-hard problem~\cite{Gavril77b,DemetrescuF03}:
\begin{center}
\begin{boxedminipage}{0.98 \columnwidth}
\textsc{Minimum Feedback Arc Set on Bounded-Degree Digraphs} ($\MFAS$) \\[5pt]
\begin{tabular}{l p{0.83 \columnwidth}}
Input: & Digraph $D = (V, A)$ whose skeleton has degree at most $3$, integer $m \le |A|$.\\
Question: & Is there a subset $A'\subseteq A$ where $|A'|\le m$ such that $D-A'$ is a DAG?
\end{tabular}
\end{boxedminipage}
\end{center}
Let $(D,m)$ be an instance of \MFAS. We construct an instance $\III$ of \BNSLaddq\ as follows:
\begin{itemize}
\item $V=V(D)$,
\item $f_y(x)=1$ for every $xy\in A(D)$,
\item $f_y(x)=0$ for every $xy \in A_V \setminus A(D)$, 
\item $\ell=|A|-m$, and
\item $q=3$.
\end{itemize}
Assume that $(D,m)$ is a \YES-instance and $A'$ is any feedback arc set of size $m$. Let $D'$ be the DAG obtained from $D$ after deleting arcs in $A'$. Then $\score(D')$ is equal to the number of arcs in $D'$, which is $|A|-m$, so $\III$ is a \YES-instance. On the other hand, if $\III$ is a \YES-instance of \BNSLadd, pick any DAG $D'$ with $\score(D')\ge \ell=|A|-m$. Without loss of generality we may assume that $A(D')\subseteq A$, as the remaining arcs have a score of zero and may hence be removed. All the arcs in $A$ have a score 1 and hence the DAG $D'$ contains at least $|A|-m$ arcs, i.e., it can be obtained from $D$ by deleting at most $m$ arcs.  Hence $(D,m)$ is also a \YES-instance. To establish the \NP-hardness of \BNSLadd, simply disregard the bound $q$ on the input.
\end{proof}
\fi

While the use of the additive representation did not affect the classical complexity of \BNSL, it makes a significant difference in terms of parameterized complexity. Indeed, in contrast to \BNSLneq:
 \iflong
\begin{theorem}
\fi
\ifshort
\begin{theorem}[]
\fi
\label{thm:additive}
\BNSLadd\ is \FPT\ when parameterized by the treewidth of the superstructure. Moreover, \BNSLaddq\ is \FPT when parameterized by $q$ plus the treewidth of the superstructure.
 \end{theorem}

\ifshort
\begin{proof}[Proof Sketch]
As noted in the preliminaries, due to space constraints we refer to the usual books for a definition of \emph{treewidth} and \emph{nice tree-decompositions}~\cite{DowneyFellows13,CyganFKLMPPS15}.
We begin by applying Bodlaender's algorithm~\cite{Bodlaender96,Kloks94} to compute a nice tree-decomposition $(\mathcal{T},\chi)$ of $G_\III$ of width $k=\tw(G_\III)$, whose vertices are called \emph{nodes}. To prove the theorem, we will design a leaf-to-root dynamic programming algorithm which will compute and store a set of records at each node of $T$, whereas once we ascertain the records for $r$ we will have the information required to output a correct answer. Intuitively, the records will store all information about each possible set of arcs between vertices in each \emph{bag}, along with relevant connectivity information provided by arcs between all vertices that are either in the current bag $t$ or in some descendant of $t$  (we denote the set of these vertices $\chi_t^\downarrow$), and information about the partial score. When solving \BNSLaddq, the records will also keep track of parent set sizes.

Formally, the records will have the following structure. For a node $t$, let $S(t)=\{(\loc,\con, \inn)~|~\loc, \con \subseteq A_{\chi(t)}, \inn:\chi(t)\to [q]_0\}$ be the set of \emph{snapshots} of $t$. The record $\RRR_t$ of $t$ is then a mapping from $S(t)$ to $\Nat_0\cup \{\bot\}$. Observe that $|S(t)|\leq 4^{k^2}(q+1)^k$.
 To introduce the semantics of our records, let 
 $\DDD_t$ be the set of all directed acyclic graphs over the vertex set $\chi_t^\downarrow$ with maximal in-degree at most $q$, and let $D_t=(\chi_t^\downarrow,A)$ be a directed acyclic graph in $\DDD_t$. We say that the \emph{snapshot of} $D_t$ in $t$ is the tuple $(\alpha,\beta,p)$ where 
$\alpha=A\cap A_{\chi(t)}$, $\beta=\Con(\chi(t), D_t)$ and $p$ (which is only used for \BNSLaddq) specifies numbers of parents of vertices from $\chi(t)$ in $D$, i.e., $p(v)=|\{w\in  \chi_t^\downarrow| wv \in A\}|$, $v\in \chi(t)$.
Now, for each snapshot $(\loc,\con,\inn)\in S(t)$ we set $\RRR_t(\loc,\con,\inn)=\bot$ if there exists no DAG in $\DDD_t$ with $(\loc, \con,\inn)$ as its snapshot, and otherwise we set $\RRR_t(\loc,\con,\inn)$ to the highest score that can be achieved by such a DAG.
       
The algorithm computes these records in a leaf-to-root fashion while traversing $\mathcal{T}$, which can be achieved in time at most $2^{\bigoh(k^2)}\cdot q^{\bigoh(k)} \cdot n$, where $n$ is the input size. Once we reach the root node $r$, we use the fact that $\chi(r)=\emptyset$ by the definition of nice tree-decompositions to simply check if $\RRR_r(\emptyset,\emptyset,\emptyset)\geq \ell$; the algorithm then outputs ``\YES'' if and only if this is the case.
\end{proof}
\fi
\iflong
\begin{proof}
We begin by proving the latter statement, and will then explain how that result can be straightforwardly adapted to obtain the former. As our initial step, we apply Bodlaender's algorithm~\cite{Bodlaender96,Kloks94} to compute a nice tree-decomposition $(\mathcal{T},\chi)$ of $G_\III$ of width $k=\tw(G_\III)$.
 In this proof we use $T$ to denote the set of nodes of $\mathcal{T}$ and $r\in T$ be the root of $\mathcal{T}$. Given a node $t\in T$, let $\chi_t^\downarrow$ be the set of all vertices occurring in bags of the rooted subtree $T_t$, i.e., $\chi_t^\downarrow=\{u~|~\exists t' \in T_t$ such that $u\in \chi(t')\}$. Let $G^\downarrow_t$ be the subgraph of $G_\III$ induced on $\chi_t^\downarrow$.

To prove the theorem, we will design a leaf-to-root dynamic programming algorithm which will compute and store a set of records at each node of $T$, whereas once we ascertain the records for $r$ we will have the information required to output a correct answer. Intuitively, the records will store all information about each possible set of arcs between vertices in each bag, along with relevant connectivity information provided by arcs between vertices in $\chi_t^\downarrow$ and information about the partial score. They will also keep track of parent set sizes in each bag. 

Formally, the records will have the following structure. For a node $t$, let $S(t)=\{(\loc,\con, \inn)~|~\loc, \con \subseteq A_{\chi(t)}, \inn:\chi(t)\to [q]_0\}$ be the set of \emph{snapshots} of $t$. The record $\RRR_t$ of $t$ is then a mapping from $S(t)$ to $\Nat_0\cup \{\bot\}$. Observe that $|S(t)|\leq 4^{k^2}(q+1)^k$.
 To introduce the semantics of our records, let 
 $\DDD_t$ be the set of all directed acyclic graphs over the vertex set $\chi_t^\downarrow$ with maximal in-degree at most $q$, and let $D_t=(\chi_t^\downarrow,A)$ be a directed acyclic graph in $\DDD_t$. We say that the \emph{snapshot of} $D_t$ in $t$ is the tuple $(\alpha,\beta,p)$ where 
$\alpha=A\cap A_{\chi(t)}$, $\beta=\Con(\chi(t), D_t)$ and $p$ specifies numbers of parents of vertices from $\chi(t)$ in $D$, i.e. $p(v)=|\{w\in  \chi_t^\downarrow| wv \in A\}|$, $v\in \chi(t)$.
We are now ready to define the record $\RRR_t$. For each snapshot $(\loc,\con,\inn)\in S(t)$:
\begin{itemize}
\item  $\RRR_t(\loc,\con,\inn)=\bot$ if and only if there exists no directed acyclic graph in $\DDD_t$ whose snapshot is $(\loc, \con,\inn)$, and
\item $\RRR_t(\loc,\con,\inn)=\tau$ if $\exists D_t\in \DDD_t$ such that 
\begin{itemize}
\item the snapshot of $D_t$ is $(\loc, \con, \inn)$, 
\item score$(D_t)=\tau$, and 
\item $\forall D'_t\in \DDD_t$ such that the snapshot of $D'_t$ is $(\loc, \con, \inn)$: score$(D_t)\geq$ score$(D'_t)$.
\end{itemize}
\end{itemize}

Recall that for the root $r\in T$, we assume $\chi(r)=\emptyset$. Hence $\RRR_r$ is a mapping from the one-element set $\{(\emptyset,\emptyset,\emptyset)\}$ to an integer $\tau$ such that $\tau$ is the maximum score that can be achieved by any DAG $D=(V,A)$ with all in-degrees of vertices upper bounded by $q$. In other words, $\III$ is a YES-instance if and only if $\RRR_r(\emptyset,\emptyset,\emptyset)\geq \ell$. To prove the theorem, it now suffices to show that the records can be computed in a leaf-to-root fashion by proceeding along the nodes of $T$. We distinguish four cases:

\textbf{$t$ is a leaf node.} Let $\chi(t)=\{v\}$. By definition, $S(t)=\{(\emptyset,\emptyset,\emptyset)\}$ and $\RRR_t(\emptyset,\emptyset,\emptyset)=f_v(\emptyset)$.

\noindent \textbf{$t$ is a forget node.} 
Let $t'$ be the child of $t$ in $\mathcal{T}$ and let $\chi(t)=\chi(t')\setminus \{v\}$. We initiate by setting $\RRR^{0}_t(\loc,\con,\inn)=\bot$ for each $(\loc,\con,\inn)\in S(t)$. 

For each $(\loc',\con',\inn')\in S(t')$, let $\loc_v$, $\con_v$ be the restrictions of $\loc'$, $\con'$ to tuples containing $v$. We now define $\loc=\loc'\setminus \loc_v$, $\con=\con' \setminus \con_v$, $\inn=\inn'|_{\chi(t)}$ and say that $(\loc,\con,\inn)$ is \emph{induced} by $(\loc',\con',\inn')$. Set $\RRR^0_t(\loc,\con,\inn):=\max(\RRR^0_t(\loc,\con,\inn),\RRR_{t'}(\loc',\con',\inn'))$, where $\bot$ is assumed to be a minimal element. 

For correctness, it will be useful to observe that $\DDD_t=\DDD_{t'}$. Consider our final computed value of $\RRR^0_t(\loc,\con,\inn)$ for some $(\loc,\con,\inn)\in S(t)$. 

If $\RRR_t(\loc,\con,\inn)=\tau$ for some $\tau\neq \bot$, then there exists a DAG $D$ which witnesses this. But then $D$ also admits a snapshot $(\loc', \con',\inn')$ at $t'$ and witnesses $\RRR_{t'}(\loc',\con',\inn')\ge \tau$.  Note that $(\loc,\con,\inn)$ is induced by $(\loc', \con', \inn')$.
So  in our algorithm $\RRR^0_t(\loc,\con,\inn)\ge \RRR_{t'}(\loc',\con',\inn')\ge \tau$.
\\\\
If on the other hand $\RRR^0_t(\loc,\con,\inn)=\tau$ for some $\tau\neq \bot$, then there exists a snapshot $(\loc',\con',\inn')$ such that $(\loc,\con,\inn)$ is induced by  $(\loc',\con',\inn')$ and $\RRR_{t'}(\loc',\con',\inn')=\tau$. $\RRR_t(\loc,\con,\inn)\ge \tau$ now follows from the existence of a DAG witnessing the value of $\RRR_{t'}(\loc',\con',\inn')$.

Hence, we can correctly set $\RRR_t=\RRR_t^0$.

\noindent \textbf{$t$ is an introduce node.} 
Let $t'$ be the child of $t$ in $\mathcal{T}$ and let $\chi(t)=\chi(t')\cup \{v\}$. We initiate by setting $\RRR^{0}_t(\loc,\con,\inn)=\bot$ for each $(\loc,\con,\inn)\in S(t)$. 

For each $(\loc',\con',\inn')\in S(t')$ and each $Q\subseteq \{ab\in A_{\chi(t)}~|~ \{a,b\}\cap \{v\}\neq \emptyset\}$, we define:
\begin{itemize}
\item $\loc:=\loc'\cup Q$ 
\item $\con:=\trcl(con' \cup Q)$
\item $\inn(x):=\inn'(x)+|\{y\in \chi(t)| yx \in Q\}|$  for every $x\in \chi(t) \setminus \{v\}$\\ $\inn(v):=|\{y\in \chi(t)| yv \in Q\}|$
\end{itemize}

 If $\con$ is not irreflexive or $\inn(x)>q$ for some $x\in \chi(t)$, discard this branch. Otherwise, let $\RRR^{0}_t(\loc,\con,\inn):=\max(\RRR^{0}_t(\loc,\con,\inn),\texttt{new})$ where $\texttt{new}=\RRR_{t'}(\loc',\con',\inn')+\sum_{ab\in Q}f_b(a)$. As before, $\bot$ is assumed to be a minimal element here.

Consider our final computed value of $\RRR^0_t(\loc,\con,\inn)$ for some $(\loc,\con,\inn)\in S(t)$.

For correctness, assume that $\RRR^0_t(\loc,\con,\inn)=\tau$ for some $\tau\neq \bot$ and is obtained from $(\loc',\con',\inn'), Q$ defined as above. Then $\RRR_{t'}(\loc',\con',\inn')=\tau-\sum_{ab\in Q} f_b(a)$. Construct a directed graph $D$ from the witness $D'$ of $\RRR_{t'}(\loc',\con',\inn')$ by adding the arcs specified in $Q$. As $\con=\trcl(con' \cup Q)$ is irreflexive and $D'$ is a DAG, $D$ is a DAG as well by $\ref{lem:con2dags}$. Moreover, $\inn(x)\leq q$ for every $x\in \chi(t)$ and the rest of vertices have in $D$ the same parents as in $D'$, so $D\in \DDD_t$. In particular, $(\loc,\con,\inn)$ is a snapshot of $D$ in $t$ and $D$ witnesses $\RRR_t(\loc,\con,\inn)\ge \RRR_{t'}(\loc',\con',\inn') + \sum_{ab\in Q} f_b(a) = \tau$.

On the other hand, if $\RRR_t(\loc,\con,\inn)=\tau$ for some $\tau\neq \bot$, then there must exist a directed acyclic graph $D=(\chi_t^\downarrow,A)$ in $\DDD_t$ that achieves a score of $\tau$. Let $Q$ be the restriction of $A$ to arcs containing $v$, and let $D'=(\chi_t^\downarrow \setminus {v},A\setminus Q)$, clearly $D'\in \DDD_{t'}$. Let $(\loc',\con',\inn')$ be the snapshot of $D'$ at $t'$. Observe that $\loc=\loc' \cup Q$, $\con=\trcl(\con' \cup Q)$, $\inn$ differs from $\inn'$ by the numbers of incoming arcs in $Q$ and the score of $D'$ is precisely equal to the score $\tau$ of $D$ minus $\sum_{(a,b)\in Q} f_b(a)$. Therefore $\RRR_{t'}(\loc',\con',\inn')\ge \tau - \sum_{(a,b)\in Q} f_b(a)$ and in the algorithm $\RRR^0_t(\loc,\con,\inn)\geq \RRR_{t'}(\loc',\con',\inn')+\sum_{(a,b)\in Q} f_b(a)\ge \tau$. Equality then follows from the previous direction of the correctness argument.

Hence, at the end of our procedure we can correctly set $\RRR_t=\RRR_t^0$.

\noindent \textbf{$t$ is a join node.} 
Let $t_1,t_2$ be the two children of $t$ in $\mathcal{T}$, recall that $\chi(t_1)=\chi(t_2)=\chi(t)$. By the well-known separation property of tree-decompositions, $\chi_{t_1}^\downarrow \cap \chi_{t_2}^\downarrow = \chi (t)$~\cite{DowneyFellows13,CyganFKLMPPS15}. 
 We initiate by setting $\RRR^{0}_t(\loc,\con,\inn):=\bot$ for each $(\loc,\con,\inn)\in S(t)$. 

Let us branch over each $\loc,\con_1,\con_2\subseteq A_{\chi(t)}$ and $\inn_1, \inn_2:\chi(t)\to [q]_0$. For every $b \in \chi(t)$ set $\inn(b)=\inn_1(b)+\inn_2(b)-|\{a|ab\in\loc\}|$. If:
\begin{itemize}
\item $\trcl(\con_1\cup \con_2)$ is not irreflexive and/or
\item $\RRR_{t_1}(\loc,\con_1,\inn_1)=\bot$, and/or
\item $\RRR_{t_2}(\loc,\con_2,\inn_2)=\bot$, and/or
\item $\inn(b)>q$ for some $b\in \chi(t)$
\end{itemize}
then discard this branch. Otherwise, set $\con=\trcl(\con_1\cup \con_2)$, $\texttt{doublecount}=\sum_{ab\in \loc} f_b(a)$ and $\texttt{new}=\RRR_{t_1}(\loc,\con_1)+\RRR_{t_2}(\loc,\con_2)-\texttt{doublecount}$. We then set $\RRR^0_t(\loc,\con,\inn):=\max(\RRR^0_t(\loc,\con,\inn),\texttt{new})$ where $\bot$ is once again assumed to be a minimal element.

At the end of this procedure, we set $\RRR_t=\RRR_t^0$.

For correctness, assume that $\RRR^0_t(\loc,\con,\inn)=\tau \neq \bot$ is obtained from $\loc,\con_1, \con_2, \inn_1, \inn_2$ as above. Let $D_1=(\chi_{t_1}^\downarrow, A_1)$ and $D_2=(\chi_{t_2}^\downarrow, A_2)$ be DAGs witnessing $\RRR_{t_1}(\loc,\con_1, \inn_1)$ and $\RRR_{t_2}(\loc,\con_2, \inn_2)$ correspondingly. 
Note that common vertices of $D_1$ and $D_2$ are precisely $\chi(t)$. In particular, if $D_1$ and $D_2$ share an 
arc $ab$, then $a,b\in \chi(t)$ and therefore $ab \in \loc$. On the other hand, $\loc \subseteq A_1$, $\loc \subseteq A_2$,  so $loc=A_1\cap A_2$. Hence $\inn$ specifies the number of parents of every $b\in\chi(T)$ in $D=D_1\cup D_2$. Rest of vertices $v\in V(D)\setminus \chi(t)$ belong to precisely one of $D_i$ and their parents in $D$ are the same as in this $D_i$. As $\trcl(\con_1\cup \con_2)$ is irreflexive, $D$ is a DAG by
Lemma~$\ref{lem:con2dags}$, so $D\in \DDD_t$. The snapshot of $D$ in $t$ is $(\loc, \con,\inn)$ and $\score(D)=\sum_{ab\in A(D)} f_b(a) =  \sum_{ab\in A_1} f_b(a) + \sum_{ab\in A_2} f_b(a) -  
 \sum_{ab\in loc} f_b(a) = \score(D_1)+\score(D_2)-\texttt{doublecount}=\RRR_{t_1}(\loc,\con_1,\inn_1)+\RRR_{t_2}(\loc,\con_2, \inn_2)-\texttt{doublecount}=\tau$. So $D$ witnesses that $\RRR_t(\loc,\con,\inn)\ge \tau$.

For the converse, assume that $\RRR_t(\loc,\con,\inn)= \tau \ne \bot $ and $D$ is a DAG witnessing this. Let $D_1$ and $D_2$ be restrictions of $D$ to $\chi_{t_1}^\downarrow$ and $\chi_{t_2}^\downarrow$ correspondingly, then by the same arguments as above $A(D_1)\cap A(D_2)=\loc$, in particular $D=D_1\cup D_2$. Let $(\loc,\con_i,\inn_i)$ be the snapshot of $D_i$ in $t_i$, $i=1,2$, then $\RRR_{t_i}(\loc,\con_i,\inn_i)\ge \score (D_i)$. By the procedure of our algorithm,
$\RRR^0_t(\loc,\con,\inn)\ge \RRR_{t_1}(\loc,\con_1,\inn_1)+\RRR_{t_2}(\loc,\con_2,\inn_2)-\texttt{doublecount} \ge \score (D_1) + \score (D_2) - \sum_{ab\in \loc} f_b(a) = \score (D) = \tau.$

Hence the resulting record $\RRR_t$ is correct, which concludes the correctness proof of the algorithm. 

Since the nice tree-decomposition $\mathcal{T}$ has $\bigoh(n)$ nodes, the runtime of the algorithm is upper-bounded by $\bigoh(n)$ times the maximum time required to process each node. This is dominated by the time required to process join nodes, for which there are at most $(2^{k^2})^3((q+1)^k)^2= 8^{k^2}\cdot (q+1)^{2k}$ branches corresponding to different choices of $\loc, \con_1, \con_2, \inn_1, \inn_2$. Constructing $\trcl(\con_1\cup \con_2)$ and verifying that it is irreflexive can be done in time $\bigoh(k^3)$. Computing $\texttt{doublecount}$ and $\inn$ takes time at most $\bigoh(k^2)$. So the record for a join node can be computed in time $2^{\bigoh(k^2)}\cdot q^{\bigoh(k)}$. Hence, after we have computed a width-optimal tree-decomposition for instance by Bodlaender's algorithm~\cite{Bodlaender96}, the total runtime of the algorithm is upper-bounded by $2^{\bigoh(k^2)}\cdot q^{\bigoh(k)} \cdot n$.

Finally, to obtain the desired result for \BNSLadd, we can simply adapt the above algorithm by disregarding the entry $\inn$ and disregard all explicit bounds on the in-degrees (e.g., in the definition of $\DDD_t$). The runtime for this dynamic programming procedure is then $2^{\bigoh(k^2)} \cdot n$.
\end{proof}
\fi

This completely resolves the parameterized complexity of \BNSLadd\ w.r.t.\ all parameters depicted on Figure~\ref{fig:overview}. However, the same is not true for \BNSLaddq: while a careful analysis of the algorithm provided in the proof of Theorem~\ref{thm:additive} reveals that \BNSLaddq\ is \XP-tractable when parameterized by the treewidth of the superstructure alone, it is not yet clear whether it is \FPT---in other words, do we need to parameterize by both $q$ and treewidth to achieve fixed-parameter tractability? 
   \ifshort
We conclude this section by answering this question affirmatively via an involved two-step reduction from a variant of the \W$[1]$-hard \textsc{Multidimensional Subset Sum} problem~\cite{GanianOS21,GanianKO21}. 

\begin{theorem}[]
\label{thm:addwhard}
\BNSLaddq\ is \W\textup{[1]}-hard when parameterized by the treewidth of the superstructure.
\end{theorem}
\fi

\iflong

We conclude this section by answering this question affirmatively. To do so, we will aim to reduce from the following problem, which can be seen as a dual to the \W$[1]$-hard \textsc{Multidimensional Subset Sum} problem considered in recent works~\cite{GanianOS21,GanianKO21}. 

\begin{center}
\begin{boxedminipage}{0.98 \columnwidth}
\textsc{Uniform Dual Multidimensional Subset Sum} (\UDMSS)\\[5pt]
\begin{tabular}{l p{0.83 \columnwidth}}
Input: & An
  integer $k$, a set $S=\{s_1,\dotsc,s_n\}$ of item-vectors with $s_i
  \in \Nat^{k}$ for every $i$ with $1\leq i \leq n$, a uniform target vector
  $t=(r,\dots,r) \in \Nat^k$, and an integer $d$.\\
Parameter: & $k$.\\
Question: & Is there a subset $S'
  \subseteq S$ with $|S'|\geq d$ such that $\sum_{s \in S'}s\leq t$?
\end{tabular}
\end{boxedminipage}
\end{center}

We first begin by showing that this variant of the problem is \W$[1]$-hard by giving a fairly direct reduction from the originally considered problem, and then show how it can be used to obtain the desired lower-bound result.

\begin{lemma}
\DMSS\ is \W$[1]$-hard.
\end{lemma}

\begin{proof}
The \W$[1]$-hard \textsc{Multidimensional Subset Sum} problem is stated as follows:
\begin{center}
\begin{boxedminipage}{0.98 \columnwidth}
\textsc{Multidimensional Subset Sum} (\MSS)\\[5pt]
\begin{tabular}{l p{0.83 \columnwidth}}
Input: & An
  integer $k$, a set $S=\{s_1,\dotsc,s_n\}$ of item-vectors with $s_i
  \in \Nat^{k}$ for every $i$ with $1\leq i \leq n$, a target vector
  $t=(t^1,\dots,t^k) \in \Nat^k$, and an integer $d$.\\
Parameter: & $k$.\\
Question: & Is there a subset $S'
  \subseteq S$ with $|S'|\leq d$ such that $\sum_{s \in S'}s\geq t$?
\end{tabular}
\end{boxedminipage}
\end{center}
Consider its dual version, obtained by reversing both inequalities:
\begin{center}
\begin{boxedminipage}{0.98 \columnwidth}
\textsc{Dual Multidimensional Subset Sum} (\DMSS)\\[5pt]
\begin{tabular}{l p{0.83 \columnwidth}}
Input: & An
  integer $k$, a set $S=\{s_1,\dotsc,s_n\}$ of item-vectors with $s_i
  \in \Nat^{k}$ for every $i$ with $1\leq i \leq n$, a target vector
  $t=(t^1,\dots,t^k) \in \Nat^k$, and an integer $d$.\\
Parameter: & $k$.\\
Question: & Is there a subset $S'
  \subseteq S$ with $|S'|\geq d$ such that $\sum_{s \in S'}s\leq t$?
\end{tabular}
\end{boxedminipage}
\end{center}
Given an instance $\III=(S,t,k,d)$ of $\MSS$, we construct an instance $\III_d=(S,z-t,k,n-d)$ of $\DMSS$, where $z=\sum_{s \in S}s$. Note that $S'$ is a witness of $\III$ if and only if $S\setminus S'$ is a witness of $\III_d$. The observation establishes $\W[1]$-hardness of $\DMSS$.   

Now it remains to show that \DMSS\ is  \W$[1]$-hard even if we require all the components of the target vector $t$ to be equal. Let $\III=(S,t,k,d)$ be the instance of \DMSS. We construct an equivalent instance $\III_\emph{eq}=(S_\emph{eq},t_\emph{eq},k+1,d+1)$ of $\UDMSS$ with $t_\emph{eq}=(d  \cdot t_\emph{max},\dots,d \cdot t_\emph{max})$, where $t_\emph{max}=\max\{t^i:i\in[k]\}$. $S_\emph{eq}$ is obtained from $S$ by setting the $(k+1)$-th entries equal to $t_\emph{max}$, plus one auxiliary vector to make the target uniform: $S_\emph{eq}=\{(a^1,\dots,a^k, t_\emph{max})|(a^1,\dots,a^k)\in S\}\cup\{b\}$, where $b=(d t_\emph{max}-t^1,\dots,d t_\emph{max}-t^k,0)$. 

For correctness, assume that $\III$ is a $\YES$-instance, in particular, we can choose $S'$ with $|S'|=d$ and $\sum_{s \in S'}s\leq t$. Then $S_\emph{eq}'=\{(a^1,\dots,a^k, t_\emph{max})|(a^1,\dots,a^k)\in S'\}\cup\{b\}$ witnesses that $\III_\emph{eq}$ is a $\YES$-instance. For the converse direction, let $\III_\emph{eq}$ be a $\YES$-instance, we choose $S'_\emph{eq}$ with $|S_\emph{eq}'|=d+1$ and $\sum_{s \in S'_\emph{eq}}\leq t_\emph{eq}$. If $b\not \in S_\emph{eq}$, sum of the $(k+1)$-th entries in $S_\emph{eq}'$ would be at least $(d+1) t_\emph{max}$, so $b$ must belong to $S'_\emph{eq}$. Then $S'_\emph{eq}\setminus \{b\}$ consists of precisely $d$ vectors with sum at most $t_\emph{eq}-b=(t^1,\dots,t^k,d t_\emph{max})$. Restrictions of these vectors to $k$ first coordinates witness that $\III$ is a $\YES$-instance.
\end{proof}
\begin{theorem}
\label{thm:addwhard}
\BNSLaddq\ is \W$[1]$-hard when parameterized by the treewidth of the superstructure.
\end{theorem}
\begin{proof}
Let $\III=(S,t,k,d)$ be an instance of $\UDMSS$ with $t=(r,\dots,r)$, and w.l.o.g.\ assume that $r$ is greater than the parameter $k$. We construct an equivalent instance $(V,\FFF, \ell, r)$ of \BNSLaddq. Let us start from the vertex set $V$. For every $i \in[k]$, we add to $V$ a vertex $v^i$ corresponding to the $i$-th coordinate of the target vector $t$. Further, for every $s=(s^1,\dots,s^k)\in S$, we add vertices $a_s, b_s$ and $s^1+\dots+s^k$ many vertices $s^i_j$, $i\in[k]$, $j\in[s^i]$. Intuitively, taking $s$ into $S'$ will correspond to adding arcs from $s^i_j$ to $v^i$ for every $i\in[k]$, $j\in[s^i]$. The upper bound $r$ for each coordinate of the sum in $S'$ is captured by allowing $v^i$ to have at most $r$ many parents. Formally, for every $s\in S$, $i\in[k]$, $j\in[s^i]$ the scores are defined as follows (for convenience we list them as scores per arc): $f(s^i_jv^i)=2$, $f(b_sa_s)=M_s=2\cdot \sum_{i\in[k]} s^i - 1$. We call the arcs mentioned so far \emph{light}. Note that for every fixed $s\in S$, $\sum_{i\in[k]} \sum_{j\in[s^i]} f(s^i_jv^i)=2\cdot \sum_{i\in[k]} s^i=M_s+1$ so the sum of scores of light arcs is $L=\sum_{s\in S}(2M_s+1)$. We finally set $f(a_s s^i_j)=f(v^ib_s)=L$ for every $s\in S$, $i\in[k]$ and $j\in[s^i]$.
Now the number of arcs yielding the score of $L$ is  $m=k|S|+\sum_{s\in S} \sum_{i\in [k]} s^i$; we call these arcs $\emph{heavy}$. We set the scores of all arcs not mentioned above to zero and we set $\ell=mL+ \sum_{s\in S}M_s + d$. This finishes our construction; see Figure~\ref{fig:reductAdditive} for an illustration. Note that the superstructure graph has treewidth of at most $k+2$: the deletion of vertices $v^i$, $i\in[k]$, makes it acyclic.

\begin{figure}[thb]
\begin{minipage}[c]{0.65\textwidth}
\vspace{-0.1cm}
\includegraphics[width= \textwidth]{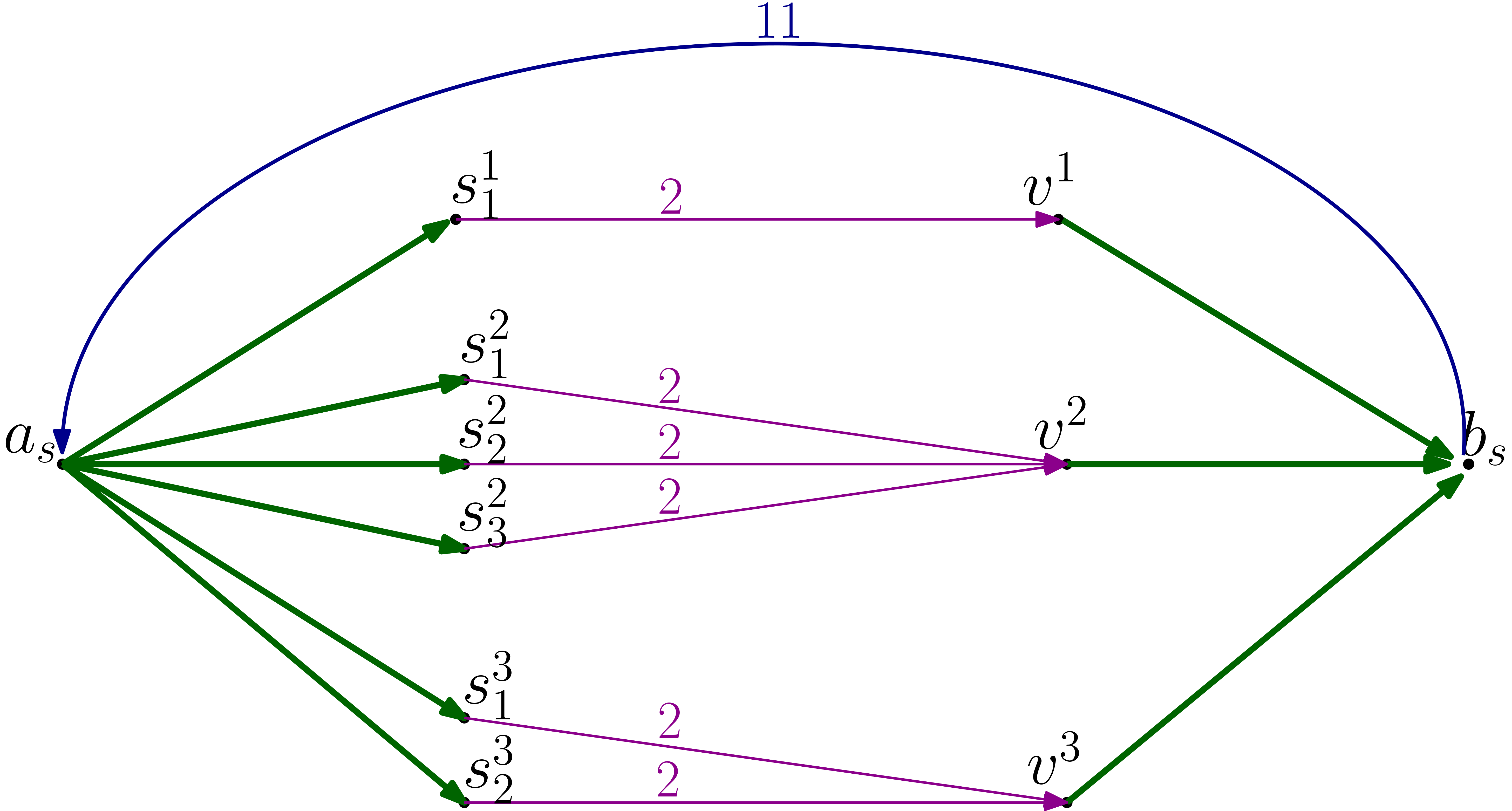}\vspace{-0.1cm}
\end{minipage}
\hfill
\begin{minipage}[c]{0.3\textwidth}
\caption{An example of our main gadget encoding the vector $s=(1,3,2)$ with $k=3$. Heavy arcs are marked in green, while purple and blue arcs are light. \\
Intuitively, the reduction forces a choice between using the blue edge or all the purple edges; the latter case provides a total score that is $1$ greater than the former, but is constrained by the upper bound $r$ on the in-degrees of $v^1$, $v^2$, $v^3$.}
\label{fig:reductAdditive}
\end{minipage}
\end{figure}

For correctness, assume that  $\III=(S,t,k,d)$ is a $\YES$-instance of $\UDMSS$, let $S'$ be a subset of $S$ of size $d$ witnessing it. We add all the heavy arcs, resulting in a total score of $mL$. Further, for every $s=(s^1,\dots,s^k)\in S'$, we add all the arcs 
$s^i_jv^i$, $i\in[k]$, $j\in[s^i]$, which increases the total score by $M_s+1$. For every $s\in S\setminus S'$, we add an arc $b_sa_s$, augmenting the total score by $M_s$. Denote the resulting digraph by $D$, then $\score(D)=mL+\sum_{s\in S'}(M_s+1)+\sum_{s\in S\setminus S'}M_s = mL+\sum_{s\in S}M_s+d=\ell$. We proceed by checking parent set sizes. Note that every $s^i_j$ has precisely one incoming arc $a_s s^i_j$ in $D$,  every $a_s$ has at most one in-neighbour $b_s$ and in-neighbours of every $b_s$ are $v^i$, $i\in[k]$. Finally, for every $i \in [k]$, $P_D(v^i)=\{s^i_j|s\in S'$, $j\in[s^i]\}$ by construction, so $|P_D(v^i)|=\sum_{s\in S'}s^i\leq r$ as $S'$ is a solution to $\UDMSS$. Therefore all the vertices in $D$ have at most $r$ in-neighbours. It remains to show acyclicity of $D$. As any cycle in the superstructure contains $v^i$ for some $i \in k$, the same holds for any potentional directed cycle $C$ in $D$. Two next vertices of $C$ after $v^i$ can be only $b_s$ and $a_s$ for some $s\in S$. In particular, by our construction, $s\in S\setminus S'$. Then, again by construction, $D$ doesn't contain an arc $s^i_jv^i$ for any $i\in[k]$, $j\in[s^i]$, so $v^i$ is not reachable from $a_s$, which contradicts to $C$ being a cycle. Therefore $D$ witnesses that $(V,\FFF, \ell, r)$ is a $\YES$-instance.

For the opposite direction, let  $(V,\FFF, \ell, r)$ be a $\YES$-instance of \BNSLaddq\ and let $D$ be a DAG witnessing this. Then $D$ contains all the heavy arcs. Indeed, sum of scores of all light arcs in $\FFF$ is $L$, so if at least one heavy arc is not in $A(D)$, then $\score(D)\leq (m-1)L+L=mL<\ell$. For every $s\in S$, let $A^s=\{s^i_jv^i| i\in[k]$, $j\in[s^i]\}$. If $D$ doesn't contain an arc $b_sa_s$ and some of arcs from $A^s$,  the total score of $A(D)\cap A^s$ is at most $M_s-1$. In this case we modify $D$ by deletion of $A(D)\cap A^s$ and addition of arc $b_sa_s$, which increases $\score (D)$ and may only decrease the parent set sizes of $v^i$, $i\in k$.  After these modifications, let $S''=\{s\in S|D$ contains an arc $b_sa_s\}$. Note that whenever $s\in S''$, $D$ cannot contain any of the arcs $s^i_jv^i$, $i\in[k]$, $j\in[s^i]$, as this would result in directed cycle $v^i\to b_s\to a_s \to s^i_j \to v^i$.  Therefore for every $s\in S$, $D$ contains either an arc $b_sa_s$ (yielding the score of $M_s$) or all of arcs $s^i_jv^i$, $i\in[k]$, $j\in[s^i]$ (yielding the score of $M_s+1$ in total), so the sum of scores of light arcs in $D$ is $\sum_{s\in S\setminus S''} (M_s+1)+ \sum_{s\in S''} M_s=\sum_{s\in S} M_s + |S\setminus S''|$, which should be at least  $\ell-mL=\sum_{s\in S}M_s + d$. So $|S\setminus S''|\geq d$, we claim that  $S'=S\setminus S''$ is a solution to $\III=(S,t,k,d)$. Indeed, for every $i\in [k]$, $r\ge |P_D(v^i)|=|\{s^i_j|s\in S'$, $j\in[s^i]\}|=\sum_{s\in S'} s^i$.
\end{proof}
\fi

\section{Implications for Polytree Learning}
\label{sec:PL}
 Here, we discuss how the results of Sections~\ref{sec:edgeparams} and~\ref{sec:add} can be adapted to \textsc{Polytree Learning} (\PL). 

\smallskip
\textbf{Theorem~\ref{thm:kernel}: Data Reduction.}\quad
Recall that the proof of Theorem~\ref{thm:kernel} used two data reduction rules. While Reduction Rule \ref{redruleone} carries over to \PLneq, Reduction Rule \ref{redruletwo} has to be completely redesigned to preserve the (non-)existence of undirected paths between $a$ and $c$. 
  By doing so, we obtain:

\iflong
\begin{theorem}
\fi
\ifshort
\begin{theorem}[]
\fi
\label{thm:kernelPL}
There is an algorithm which takes as input an instance $\III$ of \PLneq\ whose superstructure has feedback edge number $k$, runs in time $\bigoh(|\III|^2)$, and outputs an equivalent instance 
$\III'=(V',\FFF',\ell')$ of \PLneq\ such that $|V'|\leq 24k$.
\end{theorem}

\iflong
\begin{proof} Note that Reduction Rule \ref{redruleone} acts on the superstructure graph by deleting leaves and therefore preserves not only optimal scores but also (non-)existance of polytrees achiving the scores. Hence we can safely apply the rule to reduce the instance of \PLneq. After the exhaustive application, all the leaves of the superstructure graph $G$ are the endpoints of edges in feedback edge set, so there can be at most $2k$ of them. To get rid of long induced paths in $G$, we introduce the following rule:

\begin{redrule}
\label{redrulePL}
 Let $a,b_1,\dots,b_m,c$ be a path in $G$ such that for each $i\in [m]$, $b_i$ has degree precisely $2$. 
 For every $B\subseteq\{a,c\}$ and $p\in \{0,1\}$, let $\ell_p(B)$ be the maximum sum of scores that can be achieved by $b_1,\dots,b_m$ under the conditions that
(1) there exists an undirected path between $b_1$ and $b_m$ if and only if $p=1$; (2) $b_1$ (and analogously $b_m$) takes $a$ ($c$) into its parent set if and only if $a \in B$ ($c\in B$). 

We construct a new instance $\III'=(V',\FFF',\ell)$ as follows:
\begin{itemize}
\item $V':=V\cup\{b,b_1',b_1'',b_m',b_m''\}\setminus \{b_1\dots b_m\}$;
\item $\Gamma_{f'}(b_1')=\Gamma_{f'}(b_1'')=\Gamma_{f'}(b_m')=\Gamma_{f'}(b_m'')=\emptyset$;
\item The scores for $a$ (analagously $c$) are obtained from $\FFF$ by simply replacing every occurence of $b_1$ by $b_1'$ and $b_1''$ ($b_m$ by $b_m'$ and $b_m''$), formally:
\begin{itemize}
\item $\Gamma_{f'}(a)$ is a union of $\{P \in \Gamma_{f}(a)|b_1\not \in P\}$, where $f'_a(P):=f_a(P)$ and\\  $\{P\setminus b_1\cup \{b_1',b_1''\}| b_1 \in P, P \in \Gamma_{f}(a)\}$, where $f'_a(P\setminus b_1\cup \{b_1',b_1''\}):=f_a(P)$;
\item $\Gamma_{f'}(c)$ is a union of $\{P \in \Gamma_{f}(c)|b_m\not \in P\}$, where $f'_c(P):=f_c(P)$, and\\  $\{P\setminus b_m\cup \{b_m',b_m''\}| b_m \in P, P \in \Gamma_{f}(c)\}$, where $f'_c(P\setminus b_m\cup \{b_m',b_m''\}):=f_c(P)$.
\end{itemize}
\item $\Gamma_{f'}(b)$ consists of eight sets, yielding corresponding scores $f'_b$:
$\{a,c,b_1',b_1'',b_m',b_m''\}\to l_1(\{a,c\})$, 
$\{b_1',b_1'',b_m',b_m''\}\to l_0(\{a,c\})$, 
$\{b_1',b_m'\}\to l_1(\emptyset)$, 
$\emptyset \to l_0(\emptyset)$, 
$\{a,b_1',b_1'',b_m'\}\to l_1(\{a\})$, 
$\{b_1',b_1''\}\to l_0(\{a\})$, 
$\{b_m',b_m''\}\to l_1(\{c\})$,  
$\{b_1',b_m',b_m'',c\}\to l_0(\{c\})$.
\end{itemize}
\end{redrule}

Parent sets of $b$ are defined in a way to cover all the possible configurations on solutions to $\III$ restricted to $a,b_1,\dots,b_m,c$; the corresponding scores of $b$ are intuitively the sums of scores that $b_i$, $i\in[m]$, receive in the solutions. The eight cases that may arise are illustrated in Figure \ref{fig:redrulePL}.

\begin{figure}[ht]
\begin{center}
\includegraphics[width= \textwidth]{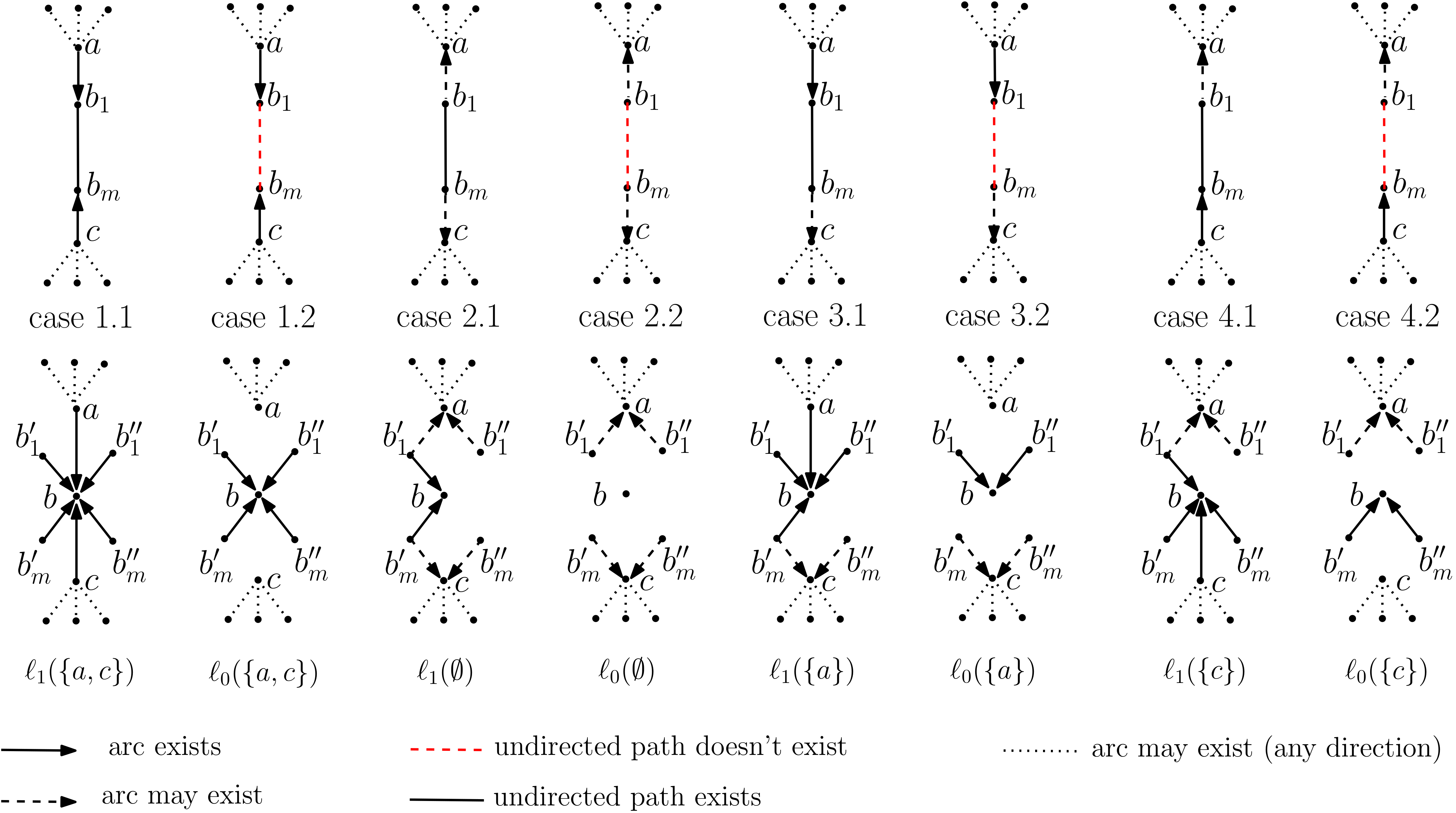}
\end{center}
\caption{Top: The eight possible scenarios for solutions to $\III$. Bottom: The corresponding arcs in the gadget after the application of Reduction Rule 2' (the scores of $b$ are specified below).
}
\label{fig:redrulePL}
\end{figure}
\begin{claim}
\label{lem:redrulePL}
Reduction Rule~\ref{redrulePL} is safe.
\end{claim}

\begin{proof}

We will show that a score of at least $\ell$ can be achieved in the original instance $\III$ if and only if a score of at least $\ell$ can be achieved in the reduced instance $\III'$.

Assume that $D$ is a polytree that achieves a score of $\ell$ in $\III$. We will construct a polytree $D'$, called the \emph{reduct} of $D$, with $f'(D')\geq \ell$. To this end, we first modify $D$ by removing the vertices $b_1,\dots,b_{m}$ and adding $b,b_1',b_1'',b_m',b_m''$. We also add arcs $b_1'a$ and $b_1''a$ ($b_m'c$ and $b_m''c$ correspondingly) if and only if $b_1a\in A(D)$ ($b_mc\in A(D)$). Let us denote the DAG obtained at this point $D^*$. Note that scores of $a$ and $c$ in $D^*$ are the same as in $D$. Further modifications of $D^*$ depend only on $D[a,b_1...b_m,c]$ and change only the parent set of $b$. We distinguish the $8$ cases listed below (see also Figure~\ref{fig:redrulePL}):
\begin{itemize}
\item case 1.1 (1.2): $ab_1,cb_m\in A(D)$, $b_1$ and $b_m$ are (not) connected by path in $D$. We add incoming arcs to $b$ from $a,c,b_1',b_1'',b_m',b_m''$ ($b_1',b_1'',b_m',b_m''$ only) resulting in $f'_b(P_{D'}(b))=l_1(\{a,c\})$ ($f'_b(P_{D'}(b))=l_0(\{a,c\})$).
\item case 2.1 (2.2): $ab_1,cb_m\not \in A(D)$, $b_1$ and $b_m$ are (not) connected by path in $D$.  We add incoming arcs to $b$ from $b_1'$ and $b_m'$ (leave $D^*$ unchanged) yielding $f'_b(P_{D'}(b))=l_1(\emptyset)$ ($f'_b(P_{D'}(b))=l_0(\emptyset)$).
\item case 3.1 (3.2):  $ab_1\in A(D)$,  $cb_m \not \in A(D)$, $b_1$ and $b_m$ are (not) connected by path in $D$.  We add incoming arcs to $b$ from $a,b_1',b_1'',b_m'$ ($b_1'$ and $b_1''$ only), then $f'_b(P_{D'}(b))=l_1(\{a\})$ ($f'_b(P_{D'}(b))=l_0(\{a\})$).
\item case 4.1 (4.2):  $ab_1 \not\in A(D)$,  $cb_m \in A(D)$, $b_1$ and $b_m$ are (not) connected by path in $D$. The cases are  symmetric to 3.1 (3.2)
\end{itemize}
Note that $D'$ contains a path between $a$ and $c$ if and only if $D$ does. By definition of $l_0$ and $l_1$, the score of $b$ in $D'$ is at least as large as the sum of scores of $b_i$, $i\in[m]$, in $D$. Moreover, each vertice in $V(D)\cap V(D')$ receives equal scores in $D$ and $D'$. Hence $D'$ is a polytree with $f'(D')\geq \ell$, as desired.
ith a score at least $\ell$.
For the converse direction, note that the polytrees constructed in cases 1.1-4.2 cover all optimal configurations which may arise in $\III'$: if there is a polytree $D''$ in $\III'$ with a score of $\ell'$, we can always modify it to a polytree $D'$ with a score of at least $\ell'$ such that  $D'[a,b_1',b_1'',b,b_m',b_m'',c]$  has one of the forms depicted at the bottom line of the figure. But every such $D'$ is a reduct of some polytree $D$ of the original instance with the same score. 
\end{proof}
We apply Reduction Rule \ref{redrulePL} exhaustively, until there is no more path to shorten. Bounds on the running time of the procedure and size of the reduced instance can be obtained similarly to the case of \BNSLneq. In particular, every long path is replaced with a set of $5$ vertices, resulting in at most $4k+4k\cdot5=24k$ vertices. 
\end{proof}
\fi

\smallskip
\textbf{Theorem~\ref{thm:lfes}: Fixed-parameter tractability.}\quad
Analagously to \BNSLneq\, a data reduction procedure as the one provided in Theorem~\ref{thm:kernelPL} does not exist for \PLneq\ parametrized by \lfes\ unless $\NP\subseteq \coNP/\cc{poly}$, since the lower-bound result provided in Theorem~\ref{thm:nokernel} can be straightforwardly adapted to \PLneq. But similarly as for \BNSL\, we can provide an \FPT\ algorithm using the same ideas as in the proof of Theorem~\ref{thm:lfes}. The algorithm proceeds by dynamic programming on the spanning tree $T$ of $G$ with $\lfes(G,T)=\lfes(G)=k$. The records will, however, need to be modified: for each vertex $v$, instead of the path-connectivity relation on $\delta(v)$, we store connected components of the \emph{inner boundary} $\delta(v)\cap V_v$ and incoming arcs to $T_v$.\ifshort~This yields:\fi
\iflong
We provide a full description of the algorithm below.

\begin{theorem}
\fi
\ifshort
\begin{theorem}[]
\fi
\label{thm:lfesPL}
\PLneq\ is fixed-parameter tractable when parameterized by the local feedback edge number of the superstructure.
\end{theorem}

\iflong
\begin{proof}
As before, given an instance $\III$ with a superstructure graph $G=G_{\III}$ such that $\lfes(G)=k$, we start from computing the spanning tree $T$ of $G$ with $\lfes(G,T)=\lfes(G)=k$; pick a root $r$ in $T$. We keep all the notations $T_v$, $V_v$, $\bar V_v$, $\delta(v)$ for $v\in V(T)$ from the subsection \ref{sub:lfes}. In addition, we define the \emph{inner boundary} of $v\in V(T)$ to be $\delta_{in}(v):=\delta(v)\cap V_v$ i.e. part of boundary that belongs to subtree of $T$ rooted in $v$. The remaining part we call the \emph{outer boundary} of $v$
and denote by $\delta_{out}(v):=\delta(v)\setminus \delta_{in}(v)$. For any set $A$ of arcs, we define $\widetilde A=\{uv|uv\in A$ or $vu\in A\}$. Obviously, the claims of Observation \ref {obs:basiclfesproperties} still hold. Moreover, for every closed $v$, $\delta_{in}(v)$ contains only $v$ itself and $\delta_{out}(v)$ is either the parent of $v$ in $T$ or $\emptyset$ (for $v=r$).

Let $R_v$ be binary relation on $\delta_{in}(v)$, $A_v\:\subseteq \delta_{out}(v)\times\delta_{in}(v)$, $s_v$ is integer. Then  
$(R_v, A_v, s_v)$ is a \emph{record} for $v$ if and only if there exist a polytree $D$ on $\bar V_v$ with all arcs oriented inside $V_v$ such that:
\begin{itemize}
\item $A_v=\{xy\in A(D)|\: x\in \delta_{out}(v), y \in \delta_{in}(v)\}$
\item $R_v=\{xy|\: x,y\in \delta_{in}(v)$ are in the same connected component of $D[V_v]\}$
\item $s_v=\sum_{u\in V_v} f_u(P_D(u))$
\end{itemize}
Note that $R_v$ is an equivalence relation on $\delta_{in}(v)$, number of its equivalence classes is equal to number of connected components of $D[V_v]$ that intersect $\delta(v)$.
 
Record $(R_v, A_v, s_v)$ is called \emph{valid} if and only if $s_v$ is maximal for fixed $R_v, A_v$ among all the records for $v$. Denote by 
$\mathcal R(v)$ the set of all valid records for $v$, then $|\mathcal R(v)|\le 2^{(2k+2)^2}$. Indeed, $R_v$ and $A_v$ can be uniquely determined by the choice of some relation on $\delta(v)\times \delta(v)$. As $|\delta(v)|\le 2k+2$, there are at most $2^{(2k+2)^2}$ possible relations. 

The root $r$ of $T$ has a single valid record $(\emptyset,\emptyset,s_{\III})$, where $s_\III$ is the maximum score that can be achieved by a solution to $\III$. For any closed $v\neq r$, $\mathcal R(v)$ consists of precisely two valid records: one for $A_v=\emptyset, R_v=\{vv\}$ and another for $A_v=\{wv\}, R_v=\{vv\}$, where $w$ is a parent of $v$ in $T$. 

We proceed by computing our records in a leaf-to-root fashion along $T$.

Let $v$ be a leaf. Start by innitiating $\mathcal R^*(v):=\emptyset$, then for each $P\in \Gamma_f(v)$ add to $\mathcal R^*(v)$ the triple $(\{vv\},\{uv|u\in P\}, f_v(P))$. Note that $\mathcal R^*(v)$ is by definition precisely the set of all records for $v$, so we can correctly set  $\mathcal R(v)=\{(R_v,A_v, s_v)\in \mathcal R^*(v)| s_v$ is maximal for fixed $R_v,A_v\}$. 

Assume that $v$ has $m$ children $\{v_i:i\in[m]\}$ in $T$, where $v_i$, $i\in[t]$, are open and $v_i$, $i\in[m]\setminus [t]$, are closed. The following claim shows how (and under which conditions) the records of children of $v$ can be composed into a record of $v$.
\begin{claim} 
 \label{lem:combRecordsPL}
Let $P\in \Gamma_f(v)$, $D_0$ is a polytree on $V_0=v\cup P$ with arc set $A_0=\{uv|u\in P\}$, $(R_i,A_i,s_i)$ are records for $v_i$ witnessed by $D_i$, $i\in[m]$. Let $A_{loc}^{in}$ be the set of arcs in $\bigcup_{i\in[t]_0}A_i$ which have both endpoints in $V_v$, $R=\trcl(\widetilde A_{loc}^{in} \cup\bigcup_{i\in[t]_0} R_i)$. Then $D=\cup_{i=0}^{m}D_i$ is a polytree if and only if the following two conditions hold:
\begin{enumerate}
         \item $A_i=\emptyset$  for each closed child $v_i \in P$.
    \item $\sum_{i=0}^t N_i - |A_{loc}^{in}| - \sum_{y \in Y} (n_y-1) = N$, where 
    \begin{itemize}
        \item $N$ is the number of equivalence classes in $\trcl(\bigcup_{i\in[t]_0} (\widetilde A_i \cup R_i))$
        \item $N_i$ is the number of equivalence classes in $R_i$, $i \in [t]$
        \item $Y$ is the set of endpoints of arcs in $\bigcup_{i\in[t]_0} A_i$ which don't belong to any $V_i$, $i \in [m]$. 
        \item For every $y\in Y$, $n_y$ is the number of arcs in  $A_0\cup...\cup A_t$ having endpoint $y$.
         
    \end{itemize}
\end{enumerate}
In this case $D$ witnesses the record $(R_v,A_v,s_v)$, where: \\$R_v=R|_{\delta_{in}(v) \times \delta_{in}(v)}$,   $A_v=(\bigcup_{i \in [t]_0}A_i)|_{\delta_{out}(v)\times\delta_{in}(v)}$,   $s_v=\sum_{i=0}^m s_i + f_v(P)$.

If $(R_v,A_v,s_v) \in \mathcal R(v)$, then $(R_i, A_i, s_i)\in \mathcal R(v_i)$, $i \in [m]$. Moreover, for any closed child $v_i\not \in P$, there is no $(R_i',A_i', s_i')\in \mathcal R(v_i)$ with $s_i'>s_i$.
\end{claim}

We will prove the claim at the end, let us show how it can be exploited to compute valid records of $v$.
We start from initial setting $\mathcal R^*(v):=\emptyset$, then branch over all parent sets $P\in \Gamma_f(v)$ and triples $(R_i,A_i, s_i)\in \mathcal R(v_i)$ for open children $v_i$. For each closed child $v_i\not \in P$ take $(R_i,A_i, s_i)\in \mathcal R(v_i)$ with maximal $s_i$, for each closed child $v_i\in P$ take $(R_i, A_i, s_i)\in \mathcal R(v_i)$ with $A_i=\emptyset$. Now the first condition of Claim\ref{lem:combRecordsPL} holds, if the second one holds as well, we add to $\mathcal R^*(v)$ the triple $(R_v, A_v,s_v)$. 

According to Claim \ref{lem:combRecordsPL}, $\mathcal R^*(v)$ computed in such a way consists only of records for $v$ and, in particular, contains all the valid records. Therefore we can correctly set $\mathcal R(v)=\{(R_v,A_v, s_v)\in \mathcal R^*(v)| s_v$ is maximal for fixed $R_v,A_v\}$.

To construct $\mathcal R^* (v)$ for node $v$ with children $v_i$, $i\in[m]$, we branch over at most $n$ possible parent sets of $v$ and at most  $2^{(2k+2)^2}$ valid records for every open child of $v$. Number of open children is bounded by $2k$, so we have at most $\bigoh((2^{(2k+2)^2})^{2k}\cdot n)  \le 2^{\bigoh(k^3)}\cdot n$ branches. In a fixed branch we compute scores for closed children in $\bigoh(n)$, application of Claim \ref{lem:combRecordsPL} requires time polynomial in $k$. So $\mathcal R^* (v)$ is computed in time  $2^{\bigoh(k^3)}\cdot n^2$ that majorizes running time for leaves. 
As the number of vertices in $T$ is at most $n$, total running time of the algorithm is $2^{\bigoh(k^3)}\cdot n^3$ assuming that $T$ is given as a part of the input.

\emph{proof of Claim \ref{lem:combRecordsPL}}
($\Leftarrow$).
We start from checking whether $D=\cup_{i=0}^m D_i$ is a polytree. As the first condition implies that a polytree of every closed child $v_i$ is connected to the rest of $D$ by at most one arc $v_iv$ or $vv_i$, it is sufficient to check whether $D^t=\cup_{i=0}^t D_i$ is polytree. Number of connected components of $D^t$ is $N'+N$, where $N'$ is the total number of connected components of $D_i$ that don't intersect $\delta(v_i)$, $i\in[t]$. Note that $D^t$ can be constructed as follows:
\begin{enumerate}
    \item Take a disjoint union of polytrees $D_i'=D_i[V_i]$, $i \in [t]_0$, then the resulting polytree has $N'+\sum_{i=0}^t N_i$ connected components.
    \item Add arcs between $D_i'$ and $D_j'$ that occur in $D$ for every $i,j \in [t]_0$, i.e. the arcs specified by $A_{loc}^{in}$. Resulting digraph is a polytree if and only if every added arc decreases the number of connected components by 1, i.e. the number of connected components after this step is $N'+\sum_{i=0}^t N_i-|A_{loc}^{in}|$.   
    \item Add all remaining vertices $y$ of $D$ together with their adjacent arcs in $D$. Note that such $y$ precisely form the set $Y$, so $D^t$ is a polytree if and only if we obtained a polytree after the previous step and every $y\in Y$ decreased it's number of connected components by $(n_y-1)$, i.e. the number $N'+N$ of connected components in $D^t$ is equal to $N'+\sum_{i=0}^t N_i-|A_{loc}^{in}|-\sum_{y \in Y} (n_y-1)$. But this is precisely the condition 2 of the claim.
\end{enumerate}
Now, assuming that $D$ is a polytree, we will show that it witnesses $(R_v, A_v,s_v)$.
Parent sets of vertices from each $V_i$ in $D$ are the same as in $D_i$, parent set of $v$ in $D$ is $P$. So $s_v=\sum_{i=0}^m s_i + f_v(P)$ is indeed the sum of scores over $V_v$ in $D$. 

There are two kinds of arcs in $D$ starting outside of $V_v$: incoming arcs to $v$ and incoming arcs to the subtrees of open children. Thus $A(D)|_{\delta_{out}(v)\times \delta_{in}(v)}=(\bigcup_{i \in [t]_0}A_i)|_{\delta_{out}(v)\times\delta_{in}(v)}=A_v$. 

Take any $u,w \in \delta_{in}(v), u\ne w$, note that $u$ and $w$ can not belong to subtrees of closed children. So $u$ and $w$ are in the same connected component of $D[V_v]$ if and only if they are connected by some undirected path $\pi$ in the skeleton of $D$ using only vertices from $D^t\cap V_v$. In this case $R_i$ captures the segmens of $\pi$ which are completely contained in $D_i[V_i], i\in [t]$. Rest of edges in $\pi$ either connect $v$ to some $V_i$, $i\in[t]$, or have enpoints in different $V_i$ and $V_j$ for some $i,j \in [t]$. Edges of this kind precisely form the set $\widetilde A_{loc}^{in}$, so $uw$ belongs to $R=\trcl(\bigcup_{i\in[t]}R_i \cup \widetilde A_{loc}^{in})$. Therefore  $R_v=R|_{\delta_{in}(v) \times \delta_{in}(v)}$ indeed represents connected components of $\delta_{in}(v)$ in $D[V_v]$.

($\Rightarrow$)
Condition 1 obviously holds, otherwise $D$ would contain a pair of arcs with the same endpoints and different directions.
In ($\Leftarrow$) we actually showed the necessity of condition 2 when 1 holds. 

For the last statement, assume that $(R_v,A_v,s_v)\in \mathcal R(v)$ but $(R_i,A_i,s_i)\not \in \mathcal R(v_i)$ for some $i$. Then there is $(R_i,A_i,s_i+\Delta) \in \mathcal R(v_i)$ for some $\delta>0$. Let $D_i'$ be a witness of $(R_i,A_i,s_i+\Delta)$, then $D'=\bigcup_{j\in[m]\setminus \{i\}}D_j \cup D_i'$ is a polytree witnessing $(R_v,A_v,s_v+\Delta)$. But this contradicts to validity of $(R_v,A_v,s_v)$. By the same arguments records for closed children $v_i\not \in P$ are the ones with maximal $s_i$ among two $(R_i,A_i,s_i)\in \mathcal R(v_i)$.  \hfill $\blacksquare$
 \end{proof}
\fi

As for treecut width, we remark that a recent reduction for \PLneq~\cite[Theorem 4.2]{GKM21} immediately implies that the problem is \W$[1]$-hard when parameterized by the treecut width\iflong (the superstructure graphs obtained in that reduction have a vertex cover of size bounded in the parameter, and the vertices outside of the vertex cover have degree at most $2$)\fi.

\smallskip
\noindent \textbf{Theorem~\ref{thm:additive}: Additive Representation.} \quad
We first show that, unlike \BNSLadd\ and \BNSLaddq, \PLaddq\ can be shown to be polynomial-time solvable for every constant value of $q$ by employing matroid-based techniques.

 \begin{theorem}
\label{thm:PLfixedq}
 Let $c$ be an arbitrary but fixed constant. Then every instance of \PLaddq\ such that $q=c$ can be solved in polynomial time.
\end{theorem}

\begin{proof}
The case of $q=1$ corresponds to computing an optimal-weight \emph{arborescence}, i.e., a spanning polytree of maximal in-degree 1, which was formulated in terms of matroid intersection and solved in polynomial time by Edmonds~\cite{EDMONDS1979}, Frank~\cite{FRANK1981} and Lawler~\cite{Lawler1976}

In fact, the same approach can be adopted for any constant bound $c$ on the in-degree. Let $G=(V,A)$ be a digraph, and $\mathcal F$ be some non-empty family of subsets of its arcs. We say that $M=(A,\mathcal F)$ is a \emph{matroid} with the \emph{ground set} $A$ and the \emph {set of independent sets} $\mathcal F$ if the following conditions hold:
\begin{itemize}
 \item for each $A'\in \mathcal F$ and $A'' \subseteq A'$ it holds that $A''\in \mathcal F$, i.e., subset of an independent set is also an independent set;
 \item if $A',B' \in \mathcal F$ and $A'$ has more elements than $B'$, then there is $a\in A' \setminus B'$ such that 
 $B' \cup \{a\} \in \mathcal F$ (this condition is often called \emph{independent set exchange property}).
\end{itemize}

A \emph {basis} of a matroid is a maximal independent set of the matroid -- that is, an independent set that is not contained in any other independent set.
A matroid $M=(A,\mathcal F)$ is called \emph{weighted} if there is a function $\omega: A \to \mathbb{R}_{\geq 0}$, called a \emph {weight function}, associating a non-negative weight to each element of the ground set. The weight function can then be naturally extended to $\mathcal F$ by setting $\omega (A') = \sum_{a\in A'} \omega (a)$. In our case, the weight function is a part of an input of \PLaddq\ .

A classical example of a matroid is so-called \emph{graphic matroid} $M_G$, where the independent set $\mathcal F$ consists of all arc sets $A'$ such that the restriction of $G$ to $A'$ is acyclic and its skeleton is a forest. Another matroid which we will use is a special case of so-called \emph{partition matroids}, i.e., matroids that impose restrictions on the in-degreees of vertices. Specifically, we consider $M_c=(A, \mathcal F_c)$, where $\mathcal F_c$ consists of all arc sets $A'$ such that the restriction of $G$ to $A'$ has maximum in-degree at most $c$. 

The crucial observation is that \PLaddq\ can be formulated as a problem of finding a common basis of these two matroids, $M_c$ and $M_G$, with the maximal weight. This problem, called \textsc{Weighted Matroid Intersection}, can be solved in polynomial time~\cite{BrezovecCG86}.
\end{proof}

\smallskip
\noindent \emph{Remark.}\quad
\emph{
The initial version of this article that appeared in the proceedings of NeurIPS 2021 erroneously claimed that \PLaddq\ is \NP-hard even when restricted to the case of $q=1$. The authors apologize for the mistake and provide the new version of Theorem~\ref{thm:PLfixedq} as a corrigendum.
}

Moreover, the dynamic programming algorithm for \BNSLaddq\ parameterized by treewidth and $q$ can be adapted to also solve \PLaddq. \iflong For completeness, we provide a full proof below; however one should keep in mind that the ideas are very similar to the proof of Theorem~\ref{thm:additive}.\fi
\ifshort
The algorithm runs in time at most $2^{\bigoh(k^2)}\cdot q^{\bigoh(k)} \cdot |\III|$.
 \fi

\iflong
\begin{theorem}
\fi
\ifshort
\begin{theorem}[]
\fi
\label{thm:additivePL}
\PLaddq\ is \FPT when parameterized by $q$ plus the treewidth of the superstructure.
\end{theorem}

\iflong
\begin{proof}
We begin by proving the latter statement, and will then explain how that result can be straightforwardly adapted to obtain the former. As our initial step, we apply Bodlaender's algorithm~\cite{Bodlaender96,Kloks94} to compute a nice tree-decomposition $(\mathcal{T},\chi)$ of $G_\III$ of width $k=\tw(G_\III)$.
 We keep the notations $T$, $r$, and $\chi_t^\downarrow$ $G^\downarrow_t$ from the proof of Theorem \ref{thm:additive}. For any arc set $A$ we denote $\widetilde A=\{uw, wu|uw\in A\}$. 

We will design a leaf-to-root dynamic programming algorithm which will compute and store a set of records at each node of $T$, whereas once we ascertain the records for $r$ we will have the information required to output a correct answer. The set of snapshots and structure of records will be the same as in the proof of Theorem \ref{thm:additive}. However, semantics wil slightly differ: in contrast to information about directed paths via forgotten nodes, $\con$ will now specify whether vertices of the bag belong to the same connected component of the partial polytree. Formally, let $\PT_t$ be the set of all polytrees over the vertex set $\chi_t^\downarrow$ with maximal in-degree at most $q$, and let $D_t=(\chi_t^\downarrow,A)$ be a polytree in $\PT_t$. We say that the \emph{snapshot of} $D_t$ in $t$ is the tuple $(\alpha,\beta,p)$ where 
$\alpha=A_{\chi(t)} \cap A $, $\beta= A_{\chi(t)} \cap \{uw|u \textup{ and }w \textup{ belong to the same connected component of }D_t\}$ and $p$ specifies numbers of parents of vertices from $\chi(t)$ in $D$, i.e. $p(v)=|\{w\in  \chi_t^\downarrow| wv \in A\}|$, $v\in \chi(t)$. We will call a connected component of $D_t$ \emph{active} if it intersects $\chi(t)$. Note that the number of equivalence classes of $\con$ is equal to the number of active connected components of $D_t$.
We are now ready to define the record $\RRR_t$. For each snapshot $(\loc,\con,\inn)\in S(t)$:
\begin{itemize}
\item  $\RRR_t(\loc,\con,\inn)=\bot$ if and only if there exists no polytree in $\PT_t$ whose snapshot is $(\loc, \con,\inn)$, and
\item $\RRR_t(\loc,\con,\inn)=\tau$ if $\exists D_t\in \PT_t$ such that 
\begin{itemize}
\item the snapshot of $D_t$ is $(\loc, \con, \inn)$, 
\item $\score(D_t)=\tau$, and 
\item $\forall D'_t\in \PT_t$ such that the snapshot of $D'_t$ is $(\loc, \con, \inn)$: $\score(D_t)\geq \score(D'_t)$.
\end{itemize}
\end{itemize}
Recall that for the root $r\in T$, we assume $\chi(r)=\emptyset$. Hence $\RRR_r$ is a mapping from the one-element set $\{(\emptyset,\emptyset,\emptyset)\}$ to an integer $\tau$ such that $\tau$ is the maximum score that can be achieved by any polytree $D=(V,A)$ with all in-degrees of vertices upper bounded by $q$. In other words, $\III$ is a YES-instance if and only if $\RRR_r(\emptyset,\emptyset,\emptyset)\geq \ell$. To prove the theorem, it now suffices to show that the records can be computed in a leaf-to-root fashion by proceeding along the nodes of $T$. We distinguish four cases:

\textbf{$t$ is a leaf node.} Let $\chi(t)=\{v\}$. By definition, $S(t)=\{(\emptyset,\emptyset,\emptyset)\}$ and $\RRR_t(\emptyset,\emptyset,\emptyset)=f_v(\emptyset)$.

\noindent \textbf{$t$ is a forget node.} 
Let $t'$ be the child of $t$ in $\mathcal{T}$ and let $\chi(t)=\chi(t')\setminus \{v\}$. We initiate by setting $\RRR^{0}_t(\loc,\con,\inn)=\bot$ for each $(\loc,\con,\inn)\in S(t)$. 

For each $(\loc',\con',\inn')\in S(t')$, let $\loc_v$, $\con_v$ be the restrictions of $\loc'$, $\con'$ to tuples containing $v$. We now define $\loc=\loc'\setminus \loc_v$, $\con=\con' \setminus \con_v$, $\inn=\inn'|_{\chi(t)}$ and set $\RRR^0_t(\loc,\con,\inn):=\max(\RRR^0_t(\loc,\con,\inn),\RRR_{t'}(\loc',\con',\inn'))$, where $\bot$ is assumed to be a minimal element. Finally we set $\RRR_t=\RRR_t^0$, correctness can be argued analogously to the case of \BNSLaddq\ .

\noindent \textbf{$t$ is an introduce node.} 
Let $t'$ be the child of $t$ in $\mathcal{T}$ and let $\chi(t)=\chi(t')\cup \{v\}$. We initiate by setting $\RRR^{0}_t(\loc,\con,\inn)=\bot$ for each $(\loc,\con,\inn)\in S(t)$. 

For each $(\loc',\con',\inn')\in S(t')$ and each $Q\subseteq \{ab\in A_{\chi(t)}~|~ \{a,b\}\cap \{v\}\neq \emptyset\}$, we define:
\begin{itemize}
\item $\loc:=\loc'\cup Q$ 
\item $\con:=\trcl(con' \cup \widetilde Q)$
\item $\inn(x):=\inn'(x)+|\{y\in \chi(t)| yx \in Q\}|$  for every $x\in \chi(t) \setminus \{v\}$\\ $\inn(v):=|\{y\in \chi(t)| yv \in Q\}|$
\end{itemize}

 Let $N$ and $N'$ be the numbers of equivalence classes in $con$ and $con'$ correspondingly. If $N\neq N'+1-|Q|$ or $\inn(x)>q$ for some $x\in \chi(t)$, discard this branch. Otherwise, let $\RRR^{0}_t(\loc,\con,\inn):=\max(\RRR^{0}_t(\loc,\con,\inn),\texttt{new})$ where $\texttt{new}=\RRR_{t'}(\loc',\con',\inn')+\sum_{ab\in Q}f_b(a)$. As before, $\bot$ is assumed to be a minimal element here.

Consider our final computed value of $\RRR^0_t(\loc,\con,\inn)$ for some $(\loc,\con,\inn)\in S(t)$.

For correctness, assume that $\RRR^0_t(\loc,\con,\inn)=\tau$ for some $\tau\neq \bot$ and is obtained from $(\loc',\con',\inn'), Q$ defined as above. Then $\RRR_{t'}(\loc',\con',\inn')=\tau-\sum_{ab\in Q} f_b(a)$. Construct a directed graph $D$ from the witness $D'$ of $\RRR_{t'}(\loc',\con',\inn')$ by adding $v$ and the arcs specified in $Q$. The equality $N=N'+1-|Q|$ garantees that every such arc decreases the number of active connected components by one, so $D$ is a polytree. Moreover, $\inn(x)\leq q$ for every $x\in \chi(t)$ and the rest of vertices have in $D$ the same parents as in $D'$, so $D\in \PT_t$. In particular, $(\loc,\con,\inn)$ is a snapshot of $D$ in $t$ and $D$ witnesses $\RRR_t(\loc,\con,\inn)\ge \RRR_{t'}(\loc',\con',\inn') + \sum_{ab\in Q} f_b(a) = \tau$.

On the other hand, if $\RRR_t(\loc,\con,\inn)=\tau$ for some $\tau\neq \bot$, then there must exist a polytree $D=(\chi_t^\downarrow,A)$ in $\PT_t$ that achieves a score of $\tau$. Let $Q$ be the restriction of $A$ to arcs containing $v$, and let $D'=(\chi_t^\downarrow \setminus {v},A\setminus Q)$, clearly $D'\in \PT_{t'}$. Let $(\loc',\con',\inn')$ be the snapshot of $D'$ at $t'$. Observe that $\loc=\loc' \cup Q$, $\con=\trcl(\con' \cup \widetilde Q)$, $\inn$ differs from $\inn'$ by the numbers of incoming arcs in $Q$ and the score of $D'$ is precisely equal to the score $\tau$ of $D$ minus $\sum_{(a,b)\in Q} f_b(a)$. Therefore $\RRR_{t'}(\loc',\con',\inn')\ge \tau - \sum_{(a,b)\in Q} f_b(a)$ and in the algorithm $\RRR^0_t(\loc,\con,\inn)\geq \RRR_{t'}(\loc',\con',\inn')+\sum_{(a,b)\in Q} f_b(a)\ge \tau$. Equality then follows from the previous direction of the correctness argument.

Hence, at the end of our procedure we can correctly set $\RRR_t=\RRR_t^0$.

\noindent \textbf{$t$ is a join node.} 
Let $t_1,t_2$ be the two children of $t$ in $\mathcal{T}$, recall that $\chi(t_1)=\chi(t_2)=\chi(t)$. 
We initiate by setting $\RRR^{0}_t(\loc,\con,\inn):=\bot$ for each $(\loc,\con,\inn)\in S(t)$. 

Let us branch over each $\loc,\con_1,\con_2\subseteq A_{\chi(t)}$ and $\inn_1, \inn_2:\chi(t)\to [q]_0$. For every $b \in \chi(t)$ set $\inn(b)=\inn_1(b)+\inn_2(b)-|\{a|ab\in\loc\}|$. Let $N_1$ and $N$ be the numbers of equivalence classes in $\con_1$ and $\trcl(\con_1\cup \con_2)$ correspondingly.  If:
\begin{itemize}
\item $\con_1\cap \con_2 \neq \trcl(\widetilde \loc)$,  and/or
\item $N-N_1\neq \frac{1}{2}|\con_2 \setminus \trcl(\widetilde \loc)|$,  and/or
\item $\RRR_{t_1}(\loc,\con_1,\inn_1)=\bot$, and/or
\item $\RRR_{t_2}(\loc,\con_2,\inn_2)=\bot$, and/or
\item $\inn(b)>q$ for some $b\in \chi(t)$
\end{itemize}
then discard this branch. Otherwise, set $\con=\trcl(\con_1\cup \con_2)$, $\texttt{doublecount}=\sum_{ab\in \loc} f_b(a)$ and $\texttt{new}=\RRR_{t_1}(\loc,\con_1)+\RRR_{t_2}(\loc,\con_2)-\texttt{doublecount}$. We then set $\RRR^0_t(\loc,\con,\inn):=\max(\RRR^0_t(\loc,\con,\inn),\texttt{new})$ where $\bot$ is once again assumed to be a minimal element.

At the end of this procedure, we set $\RRR_t=\RRR_t^0$.

For correctness, assume that $\RRR^0_t(\loc,\con,\inn)=\tau \neq \bot$ is obtained from $\loc,\con_1, \con_2, \inn_1, \inn_2$ as above. Let $D_1=(\chi_{t_1}^\downarrow, A_1)$ and $D_2=(\chi_{t_2}^\downarrow, A_2)$ be polytrees witnessing $\RRR_{t_1}(\loc,\con_1, \inn_1)$ and $\RRR_{t_2}(\loc,\con_2, \inn_2)$ correspondingly. 
Recall from the proof of Theorem \ref{thm:additive} that common vertices of $D_1$ and $D_2$ are precisely $\chi(t)$, $loc=A_1\cap A_2$ and $\inn$ specifies the number of parents of every $b\in\chi(T)$ in $D=D_1\cup D_2$. Numbers of active connected components of $D$ and $D_1$ are $N$ and $N_1$ correspondingly. Observe that $D$ can be constructed from $D_1$ by adding vertices and arcs of $D_2$. As $\con_1\cap \con_2 = \trcl(\widetilde \loc)$, we can only add a path between vertices in $\chi(t)$ if it didn't exist in $D_1$. Hence $\frac{1}{2}|\con_2 \setminus \trcl(\widetilde \loc)|$ specifies the number of paths between vertices in $\chi(t)$ via forgotten vertices of $\chi_{t_2}^\downarrow$. The equality $N_1-N=\frac{1}{2}|\con_2 \setminus \trcl(\widetilde \loc)|$ means that adding every such path decreases the number of active connected components of $D_1$  by one. As $D_1$ is a polytree, $D$ is a polytree as well, so $D\in \PT_t$. The snapshot of $D$ in $t$ is $(\loc, \con,\inn)$ and $\score(D)=\sum_{ab\in A(D)} f_b(a) =  \sum_{ab\in A_1} f_b(a) + \sum_{ab\in A_2} f_b(a) -  
 \sum_{ab\in loc} f_b(a) = \score(D_1)+\score(D_2)-\texttt{doublecount}=\RRR_{t_1}(\loc,\con_1,\inn_1)+\RRR_{t_2}(\loc,\con_2, \inn_2)-\texttt{doublecount}=\tau$. So $D$ witnesses that $\RRR_t(\loc,\con,\inn)\ge \tau$.

For the converse, assume that $\RRR_t(\loc,\con,\inn)= \tau \ne \bot $ and $D$ is a polytree witnessing this. Let $D_1$ and $D_2$ be restrictions of $D$ to $\chi_{t_1}^\downarrow$ and $\chi_{t_2}^\downarrow$ correspondingly, then $A(D_1)\cap A(D_2)=\loc$, in particular $D=D_1\cup D_2$. Let $(\loc,\con_i,\inn_i)$ be the snapshot of $D_i$ in $t_i$, $i=1,2$. $D=D_1\cup D_2$ is a polytree, so any pair of vertices in $\chi(t)$ can not be connected by different paths in $D_1$ and $D_2$, i.e. $\con_1\cap \con_2 = \trcl(\widetilde \loc)$. By the procedure of our algorithm,
$\RRR^0_t(\loc,\con,\inn)\ge \RRR_{t_1}(\loc,\con_1,\inn_1)+\RRR_{t_2}(\loc,\con_2,\inn_2)-\texttt{doublecount} \ge \score (D_1) + \score (D_2) - \sum_{ab\in \loc} f_b(a) = \score (D) = \tau.$

Hence the resulting record $\RRR_t$ is correct, which concludes the correctness proof of the algorithm. 

Since the nice tree-decomposition $\mathcal{T}$ has $\bigoh(n)$ nodes, the runtime of the algorithm is upper-bounded by $\bigoh(n)$ times the maximum time required to process each node. This is dominated by the time required to process join nodes, for which there are at most $(2^{k^2})^3((q+1)^k)^2= 8^{k^2}\cdot (q+1)^{2k}$ branches corresponding to different choices of $\loc, \con_1, \con_2, \inn_1, \inn_2$. Constructing $\trcl(\con_1\cup \con_2)$ and computing numbers of active connected components can be done in time $\bigoh(k^3)$. Computing $\texttt{doublecount}$ and $\inn$ takes time at most $\bigoh(k^2)$. So the record for a join node can be computed in time $2^{\bigoh(k^2)}\cdot q^{\bigoh(k)}$. Hence, after we have computed a width-optimal tree-decomposition for instance by Bodlaender's algorithm~\cite{Bodlaender96}, the total runtime of the algorithm is upper-bounded by $2^{\bigoh(k^2)}\cdot q^{\bigoh(k)} \cdot n$.

\end{proof}
\fi

The situation is, however, completely different for \PLadd: unlike \BNSLadd, this problem is in fact polynomial-time tractable. Indeed, it admits a simple reduction to the classical minimum edge-weighted spanning tree problem.

\iflong
\begin{observation}
\fi
\ifshort
\begin{observation}[]
\fi
\label{obs:PLeasy}
\PLadd\ is polynomial-time tractable.
\end{observation}

\iflong
\begin{proof}
Consider an the superstructure graph $G$ of an instance $\III=(V,\FFF,\ell)$ of \PLadd\ where we assign to each edge $ab\in E(G)$ a weight $w(ab)=\max{f_a(b), f_b(a)}$, and recall that we can assume w.l.o.g. that $G$ is connected. Each spanning tree $T$ of $G$ with weight $p$ can be transformed to a DAG $D$ over $V$ with a score of $p$ and whose skeleton is a tree by simply replacing each edge $ab$ with the arc $ab$ or $ba$, depending on which achieves a higher score. On the other hand, each solution to $\III$ can be transformed into a spanning tree $T$ of the same score by reversing this process. The claim then follows from the fact that a minimum-weight spanning tree of a graph can be computed in time $\bigoh(|V|\cdot \log |V|)$.
\end{proof}
\fi

This coincides with the intuitive expectation that learning simple, more restricted networks could be easier than learning general networks. We conclude our exposition with an example showcasing that this is not true in general when comparing \PL\ to \BNSL. Grüttemeier et al.~\cite{GKM21} recently showed that \PLneq\ is \W$[1]$-hard when parameterized by the number of \emph{dependent vertices}, which are vertices with non-empty sets of candidate parents in the non-zero representation. For \BNSLneq\ we can show:

\iflong
\begin{theorem}
\fi
\ifshort
\begin{theorem}[]
\fi
\label{thm:depfpt}
\BNSLneq\ is fixed-parameter tractable when parameterized by the number of dependent vertices.
\end{theorem}

\iflong
\begin{proof}
Consider an algorithm $\mathbb{B}$ for \BNSLneq\ which proceeds as follows. First, it identifies the set $X$ of dependent vertices in the input instance $\III=(V,\FFF,\ell)$, and then it branches over all possible choices of arcs with both endpoints in $X$, i.e., it branches over each arc set $A\subseteq A_X$. This results in at most $3^{k^2}$ branches, where $k=|X|$. In each branch and for each vertex $x\in X$, it now finds the highest-scoring parent set among those which precisely match $A$ on $X$, i.e., it first computes $\parentsets^A(x)=\SB P\in parentsets(x)\SM \forall w\in X\setminus \{x\}: w\in P \iff wp\in A \SE$ and then computes $\score^A(x)=\max_{P\in \parentsets^A(x)}(f_x(P))$. It then compares $\sum_{x\in X}\score^A(x)$ to $\ell$; if the former is at least as large as the latter in at least one branch then $\mathbb{B}$ outputs ``\YES'', and otherwise it outputs no.

The runtime of this algorithm is upper-bounded by $\bigoh(3^{k^2}\cdot k\cdot |\III|)$. As for correctness, if $\III$ admits a solution $D$ then we can construct a branch such that $\mathbb{B}$ will output ``\YES'': in particular, this must occur when $A$ is equal to the arcs of the subgraph of $D$ induced on $X$. On the other hand, if $\mathbb{B}$ outputs ``\YES'' for some choice of $A$, we can construct a DAG $D$ with a score of at least $\ell$ by extending $A$ as follows: for each $x\in X$ we choose a parent set $P\in \parentsets^A(x)$ which maximizes $f_x(P)$ and we add arcs from each vertex in $P\setminus X$ to $x$. The score of this DAG will be precisely $\sum_{x\in X}\score^A(x)$, which concludes the proof.
\end{proof}
\fi	

\section{Concluding Remarks}

Our results provide a new set of tractability results that counterbalance the previously established algorithmic lower bounds for \textsc{Bayesian Network Structure Learning} and \textsc{Polytree Learning} on ``simple'' superstructures. In particular, even though the problems remain \W$[1]$-hard when parameterized by the vertex cover number of the superstructure~\cite{OrdyniakS13,GKM21}, we obtained fixed-parameter tractability and a data reduction procedure using the feedback edge number and its localized version. Together with our lower-bound result for treecut width, this completes the complexity map for \BNSL\ w.r.t.\ virtually all commonly considered graph parameters of the superstructure. Moreover, we showed that if the input is provided with an additive representation instead of the non-zero representation considered in previous theoretical works, the problems admit a dynamic programming algorithm which guarantees fixed-parameter tractability w.r.t.\ the treewidth of the superstructure. We remark that all of our results assume that the score functions are provided explicitly; future work could also consider the behavior of the problem when these functions are supplied by a suitably defined oracle.

This theoretical work follows up on previous complexity studies of the considered problems, and as such we do not claim any immediate practical applications of the results. That being said, it would be interesting to see if the polynomial-time data reduction procedure introduced in Theorem~\ref{thm:kernel} could be adapted and streamlined 
\iflong (and perhaps combined with other reduction rules which do not provide a theoretical benefit, but perform well heuristically)  \fi
to allow for a speedup of previously introduced heuristics for the problem~\cite{ScanagattaCCZ15,ScanagattaCCZ16}, at least for some sets of instances.
\ifshort
Finally, we believe that the \emph{local feedback edge number} can be used to push the boundaries of tractability for other problems of interest as well.
\fi

\iflong
Last but not least, we'd like to draw attention to the \emph{local feedback edge number} parameter introduced in this manuscript specifically to tackle \BNSL. This generalization of the feedback edge set has not yet been considered in graph-theoretic works; while it is similar in spirit to the recent push towards measuring the so-called \emph{elimination distance} of a graph to a target class, it is not captured by that notion. Crucially, we believe that the applications of this parameter go beyond \BNSL; all indications suggest that it may be used to achieve tractability also for purely graph-theoretic problems where previously only tractability w.r.t. \fes\ was known.
\fi

\medskip
\noindent
\textbf{Acknowledgments.}\quad The authors acknowledge support by the Austrian Science Fund (FWF, projects P31336 and Y1329). 

Moreover, the authors wish to thank Juha Harviainen, Frank Sommer and Manuel Sorge for informing them of the mistake in the proof of Theorem~\ref{thm:PLfixedq} in the extended abstract.

\bibliographystyle{plain}
\bibliography{references}

\end{document}